\documentclass[12pt]{article}

\usepackage{bm}
\usepackage{amsmath,color}
\usepackage{amssymb,enumerate}
\usepackage{epsfig,graphicx}
\usepackage{amsthm}
\usepackage{booktabs}
\usepackage{caption}

\usepackage{amsmath}
\usepackage{natbib}
\usepackage{url} 

\usepackage[ruled,linesnumbered]{algorithm2e}
\usepackage{graphicx,epstopdf}
\usepackage[monochrome]{xcolor}

\newtheorem{corollary}{Corollary}[section]
\newtheorem{theorem}{Theorem}[section]

\newtheorem{assumption}{Assumption}

\newtheorem{lemma}{Lemma}

\renewcommand{\theenumii}{\theenumi.\arabic{enumii}.}

\newcommand{\blind}{0}

\begin{document}

\def\spacingset#1{\renewcommand{\baselinestretch}%
{#1}\small\normalsize} \spacingset{1}


\if0\blind
{
  \title{\bf Stability Approach to Regularization Selection for Reduced-Rank Regression}
    \author{Canhong Wen, Qin Wang\\
    International Institute of Finance, School of Management,\\ University of Science and Technology of China\\\\
 	Yuan Jiang\\
 	Department of Statistics, Oregon State University\\}
  \maketitle
} \fi

\if1\blind
{
  \bigskip
  \bigskip
  \bigskip
  \begin{center}
    {\LARGE\bf Stability approach to regularization selection for reduced-rank regression}
\end{center}
  \medskip
} \fi

\begin{abstract}
The reduced-rank regression model is a popular model to deal with multivariate response and multiple predictors, and is widely used in biology, chemometrics, econometrics,  engineering, and other fields. In the reduced-rank regression modelling, a central objective is to estimate the rank of the coefficient matrix that represents the number of effective latent factors in predicting the multivariate response. Although theoretical results such as rank estimation consistency have been established for various methods, in practice rank determination still relies on information criterion based methods such as AIC and BIC or subsampling based methods such as cross validation. Unfortunately, the theoretical properties of these practical methods are largely unknown. In this paper, we present a novel method called StARS-RRR that selects the tuning parameter and then estimates the rank of the coefficient matrix for reduced-rank regression based on the stability approach. We prove that StARS-RRR achieves rank estimation consistency, i.e., the rank estimated with the tuning parameter selected by StARS-RRR is consistent to the true rank. Through a simulation study, we show that StARS-RRR outperforms other tuning parameter selection methods including AIC, BIC, and cross validation as it provides the most accurate estimated rank. In addition, when applied to a breast cancer dataset, StARS-RRR discovers a reasonable number of genetic pathways that affect the DNA copy number variations and results in a smaller prediction error than the other methods with a random-splitting process.
\end{abstract}


\spacingset{1.5} 

\section{Introduction}

To model the relationship between two sets of multivariate data is of great significance in both theory and practice. One commonly used model for two sets of multivariate data is the multivariate regression model, $y = C x + e$, which assumes a linear relationship between the multivariate predictor $ x $ and multivariate response $y$, with a coefficient matrix $C$ and random errors $e$. An additional natural assumption of this model is that the coefficient matrix $C$ has a low-rank structure, which further defines the reduced-rank regression model
\begin{equation}
y = Cx + e,\quad \text{rank}(C) \le k,
\end{equation}
where $k$ is a given integer. Since \citet{anderson1951estimating} proposed the reduced-rank regression model, it has been widely used in biology, econometrics, image science and many other fields \citep{anderson2002reduced,kobak2019sparse,wen2020co}.

Many methods have been proposed to estimate the coefficient matrix in reduced-rank regression. Among them, \citet{bunea2011optimal} proposed a penalized least squares estimator of $C$ with the $l_0$-norm penalty of the singular values of $C$, based on the fact that the rank of a matrix is equal to the number of its non-zero singular values. Motivated by the application of the $l_1$-norm penalty to variable selection, \citet{yuan2007dimension} proposed a penalized least squares estimator of $C$ with the nuclear norm penalty of the coefficient matrix, which is the $l_1$-norm of the singular values of $C$. To reduce estimation bias caused by the nuclear norm penalty, \citet{chen2013reduced} introduced an adaptive nuclear norm penalty and showed that their estimator achieves a flexible bias-variance trade-off: a large singular value receives a small penalty to control bias, and a small singular value receives a large penalty to induce sparsity. Other reduced-rank regression methods include sparse reduced-rank regression that further imposes sparsity on the columns of the coefficient matrix \citep{chen2012sparse} and co-sparse reduced-rank regression that imposes sparsity on both the rows and columns of the coefficient matrix \citep{wen2020co}.

In the study of low-rank matrix estimation, rank determination has always been a key issue \citep{kanagal2010rank,ashraphijuo2017rank,kong2020a}. The influence of rank determination on coefficient estimation has been studied by \citet{anderson2002specification}: when the rank is underestimated, it will lead to estimation bias; when the rank is overestimated, variance of the estimator can be unnecessarily large. To effectively determine the rank of the coefficient matrix, all the aforementioned methods rely on tuning parameters that need to be selected. In fact, theoretical results such as rank estimation consistency have been established  \citep{bunea2011optimal,chen2013reduced} when the tuning parameters are chosen to meet certain theoretical conditions. However, these conditions usually involve unknown model parameters and thus are hard to verify in real applications. In practice, it still relies heavily on empirical methods to select tuning parameters, such as information criteria and subsampling methods.

The most well-known information criteria are probably the Akaike Information Criterion (AIC, \citealt{akaike1974a}) and the Bayesian Information Criterion (BIC, \citealt{schwars1978estimating}). They are both widely used in reduced-rank regression \citep{corander2004bayesian,chen2012reduced,bernardini2015macroeconomic}. In addition, generalized cross validation (GCV) and generalized information criterion (GIC), proposed by \citet{golub1979generalized} and \citet{fan2013tuning} respectively for linear models, are also often used in reduced-rank regression \citep{yuan2007dimension, chen2016model}. {\color{blue}\citet{she2017selective} considered a selective reduced-rank regression model that possesses both low-rank and sparse structure, and proposed a predictive information criterion (PIC) for tuning parameter selection. This method can also be applied to reduced-rank regression.} Other information criteria include an extension of BIC for high-dimensional data named BICP \citep{an2008stepwise}. The subsampling methods are represented by cross validation, which has also been widely used in variable selection and reduced-rank regression \citep{mukherjee2011reduced,ulfarsson2013tuning,jiang2016variable}. Although these methods have been repeatedly used, their theoretical properties are largely unknown and still need to be studied. One such effort is \citet{she2019on} that investigated cross validation for sparse reduced-rank regression and found that conventional cross validation may be associated with inconsistent models on different training sets. Instead of cross validating the tuning parameter, \citet{she2019on} proposed to cross validate the sparsity structure to maintain the same model in different trainings and validations.

Another type of subsampling methods for high-dimensional data is called stability approach. Different from cross validation and generalized cross validation that evaluate the prediction accuracy of a model, stability approaches focus on the stability of a model across subsamples. Stability selection \citep{meinshausen2010stability} is the first stability approach proposed for high-dimensional data and has a wide range of applicability, such as variable selection, graphical modelling, and cluster analysis. Furthermore, \citet{liu2010stability} proposed the Stability Approach to Regularization Selection (StARS) method to select the tuning parameter in graphical models. Compared to other tuning parameter selection methods such as AIC, BIC, and cross validation, StARS enjoys both theoretical and empirical advantages in graphical models. Further studies of stability approaches include \citet{shah2013variable}, \citet{yu2013stability}, \citet{sun2013consistent}, among others. Although the stability approach has been shown powerful for variable selection and graphical models, its theoretical and empirical performances have yet to be studied in reduced-rank regression. This is the main motivation of our work.

In this article, we propose a new tuning parameter selection method for reduced-rank regression based on stability approach. To appropriately apply the idea of stability approach, we define a new concept of instability specially for the reduced-rank regression models, that is, the sample variance of the estimated rank from the subsamples. In adjunct with the newly defined instability, we propose a new algorithm to select the tuning parameter based on the behavior of the instability along a grid of increasing tuning parameters. We call the new method the Stability Approach to Regularization Selection for 
Reduced-Rank Regression (StARS-RRR). Theoretically, we establish the consistency of rank estimation for StARS-RRR: the estimated rank is equal to the true rank with probability tending to one, a result that is stronger than the partial sparsistency property established for StARS \citep{liu2010stability}. Empirically, we show that StARS-RRR outperforms information criteria and other subsampling methods for both simulated and real data. In simulated data, StARS-RRR recovers the rank correctly in most replications and leads to the smallest bias as long as the signal to noise ratio is not extremely small; in real data, StARS-RRR discovers a reasonable number of relationships between copy number variations and gene expressions among breast cancer patients and results in a smaller prediction error with a random-splitting process than the other methods.

\section{Methodology and Algorithm}\label{secrrr}

\subsection{Adaptive nuclear norm penalization}

In the following subsections, we will introduce the methodology and algorithm for StARS-RRR. As seen later, StARS-RRR is a general framework that can be applied to any reduced-rank regression method. For the purpose of illustration, we will use the adaptive nuclear norm penalization method \citep{chen2013reduced} as an example to introduce StARS-RRR.

Denote $Y=(Y_1,\ldots, Y_n)^T \in \mathbb{R}^{n\times q}$ as the response, $X=(X_1,\ldots, X_n)^T \in \mathbb{R}^{n\times p}$ as the design matrix, and assume that they follow a multivariate linear model
\begin{equation}
Y = X C + E, \label{eqn:mlm}
\end{equation}
where $C \in \mathbb{R}^{p\times q}$ is the coefficient matrix that is often assumed to have a low rank and $E \in \mathbb{R}^{n\times q}$ is the error matrix. The adaptive nuclear norm penalization method aims to estimate the coefficient matrix $C$ by considering the following optimization problem
\begin{equation}
\hat{C}_{\lambda}=\mathop{\arg\min}\limits_{C\in\mathbb{R}^{p\times q}}\left\{\frac{1}{2}\lVert Y-XC\rVert_{F}^{2} + \lambda\lVert XC\rVert_{\ast w}\right\}. \label{eqn:ann}
\end{equation}
In \eqref{eqn:ann}, $\|\cdot\|_F$ denotes the Frobenius norm and $ \| XC \|_{\ast w}=\sum_{i=1}^{p\wedge q}w_{i}d_{i}(XC) $ denotes the adaptive nuclear norm of the matrix $ XC $ with $ w_{i}=d_{i}^{-\gamma}(PY)$, where $ d_{i}(M)$ denotes the $i$-th largest singular value of a matrix $M$, $P = X (X^T X)^- X^T$, and $\gamma\ge 0 $. The above optimization leads to an explicit form for the rank of the estimated coefficient matrix $\hat{C}_{\lambda}$ \citep{chen2013reduced}:
\begin{equation}
\hat{r}_{\lambda}= \max\{r: d_r(PY) > \lambda^{1/(\gamma + 1)}\}.\label{equ:est-r}
\end{equation}

Based on the above explicit form, it was shown that the adaptive nuclear norm penalization method recovers the true rank of $C$ with high probability if the error matrix  $E$ has independent $N(0,\sigma^2)$ entries and the tuning parameter $\lambda$ satisfies the condition $\lambda=\left\{(1+\theta) \left(\sqrt{r_\mathrm{x}}+\sqrt{q}\right) \sigma/ \delta\right\}^{\gamma+1}$ for some $\theta>0$, where $r_\mathrm{x}$ is the rank of $X$ and $\delta$ is a constant depending on the singular values of $XC$. However, this theoretical result may not be used to select $\lambda$ in practice because it involves the unknown parameters $\sigma$ and $\delta$. Thus, it is essential to use a practical and data-driven procedure to select the tuning parameter $\lambda$.

\subsection{Tuning parameter selection for reduced-rank regression}

In this subsection, we will review existing methods to select tuning parameters for reduced-rank regression. Roughly, there are two types of tuning parameter selection methods, one based on information criteria and the other based on subsampling technique. Given a tuning parameter $\lambda$, denote the estimator of \eqref{eqn:ann} as $ {\hat C}_{\lambda} $, and the corresponding rank as $\hat r_{\lambda}$. For an information criterion based method, one chooses an optimal $\lambda$ by minimizing a function of the sum of squared error and the degree of freedom. Table~\ref{tab:loss} shows some widely used information criteria.

\iftrue
\begin{table}
	\caption{\label{tab:loss}Information criterion based methods for tuning parameter selection in \eqref{eqn:ann}.}
	\centering
	\begin{tabular}{ll}
		\hline
		Information Criterion   &  Mathematical Formula  \\
		\hline
		AIC    &    $ nq\log\left(\dfrac{\|Y- X {\hat C}_{\lambda} \|_{\mathrm{F}}^2}{nq}\right)+2\hat r_{\lambda}(r_\mathrm{x}+q-\hat r_{\lambda})  $ \\
		BIC    &   $ nq\log\left(\dfrac{\|Y- X {\hat C}_{\lambda} \|_{\mathrm{F}}^2}{nq}\right)+\log(nq)\hat r_{\lambda}(r_\mathrm{x}+q-\hat r_{\lambda}) $ \\
		GIC    &   $ nq\log\left(\dfrac{\|Y- X {\hat C}_{\lambda}\|_{\mathrm{F}}^2}{nq}\right)+\log\log(nq)\log(pq)\hat r_{\lambda}(r_\mathrm{x}+q-\hat r_{\lambda}) $ \\
		BICP   &   $ nq\log\left(\dfrac{\|Y- X {\hat C}_{\lambda} \|_{\mathrm{F}}^2}{nq}\right)+2\log(pq)\hat r_{\lambda}(r_\mathrm{x}+q-\hat r_{\lambda}) $ \\
		GCV    &   $ \dfrac{nq \|Y- X {\hat C}_{\lambda} \|_{\mathrm{F}}^2}{[nq-\hat r_{\lambda}(r_\mathrm{x}+q-\hat r_{\lambda})]^2} $  \\
		{\color{blue}PIC}    &   {\color{blue}$\dfrac{\| Y-X\hat C_\lambda\|^2_F}{nq-2\hat r_\lambda(r_X+q-\hat r_\lambda)}$}  \\
		\hline
	\end{tabular}
\end{table}
\fi

Information criteria play an important role in model selection and model assessment, generally composed of an error term and a degree of freedom term to balance the goodness of fit and the model complexity. They are a popular choice for tuning parameter selection in reduced-rank regression, yet they may not be able to accurately estimate the rank when the true rank is low \citep{velu2013multivariate}. Information criteria like AIC and BIC tend to overestimate the rank and their performance deteriorates with small sample size. In addition to AIC and BIC, other information criterion such as GCV, GIC, and BICP have also been used to estimate rank in reduced-rank regression \citep{yuan2007dimension, chen2016model}, since they were initially proposed for linear models. {\color{blue}Proposed for selective reduced-rank regression models with both low-rank and sparse structures, PIC can also be simplified for reduced-rank regression models to the form in Table~\ref{tab:loss}. For PIC, \citet{she2017selective} established a non-asymptotic oracle inequality for its estimation error rate without assuming an infinite sample size.}

The subsampling methods include cross validation, which is also one of the most popular methods applied in reduced-rank regression \citep{ulfarsson2013tuning, kobak2019sparse}. $K$-fold cross validation randomly divides the data set into $K$ subsets of the same size. Each time one subset is reserved as the test set and the other $K-1$ subsets serve as the training set. Then, the model with a given tuning parameter is trained on the training set and tested on the test set. Last, the average test error of the $K$ models is taken as the cross validation score. The tuning parameter with the smallest cross validation score is taken as the selected parameter. \citet{she2019on} investigated cross validation for sparse reduced-rank regression and found that conventional cross validation may be associated with inconsistent models on different training sets. Instead of cross validating the tuning parameter, \citet{she2019on} proposed to cross validate the sparsity structure to maintain the same model in different trainings and validations.

Although information criteria and subsampling methods have been widely used in reduced-rank regression, their theoretical properties are largely unknown and still need to be studied. A tuning parameter selection method for reduced-rank regression that has a solid theoretical property is imperative for practical use.

\subsection{StARS-RRR} \label{sec:sta}

We hereby propose a novel tuning parameter selection method called StARS-RRR for reduced-rank regression based on the stability approach. The stability approach was originally introduced for variable selection and graphical models \citep{meinshausen2010stability, liu2010stability}. In particular, we borrow the idea from StARS \citep{liu2010stability}. The key ingredient of StARS is to define a measure called instability of an estimator from randomly drawn subsamples of the original data. For instance, \citet{liu2010stability} proposed to use the variance of the Bernoulli indicator of an edge in a graph, averaged over all edges and all estimated graphs from randomly drawn subsamples.

Before presenting the definition of instability for StARS-RRR, we introduce some notation. Let $b = b(n)$ be such that $1 < b(n) < n$. We draw $N$ random subsamples $S_1,\ldots,S_N$ from $\{(X_1,Y_1),\ldots,(X_n,Y_n)\}$ without replacement, each of size $b$. Theoretically, there are $N = \binom{n}{b}$ such subsamples. However, \citet{politis1999subsampling} argues that it suffices in practice to choose a large number $N$ of subsamples at random. Based on the $i$-th subsample $S_i$, we derive the adaptive nuclear norm penalization estimator, $\hat{C}_\lambda(S_i)$, from \eqref{eqn:ann}, with rank $\hat{r}_{\lambda}(S_i)$ for a pre-specified tuning parameter $\lambda > 0$.

As our objective is to select a tuning parameter $\lambda$ so that the estimated rank of $C$ is close to the true rank, a natural definition of instability arises from the variation of the estimated ranks from the subsamples. Therefore, we define the instability corresponding to $\lambda$ as the sample variance of $\{\hat{r}_{\lambda}(S_i): i = 1,\ldots,N\}$:
\begin{equation}\label{eqn:ins}
\hat{D}(\lambda) = S^{2}(\hat{r}_{\lambda})  = \frac{1}{N-1} \left[ \sum_{i=1}^N \hat{r}^2_{\lambda}(S_i) - N \left\{\frac1N \sum_{i=1}^N \hat{r}_{\lambda}(S_i) \right\}^2 \right].
\end{equation}

To see how the instability varies for different values of $\lambda$, we demonstrate via a simulated dataset. The dataset was simulated under the setting that is detailed in Section \ref{sec:sim}, where we note that the true rank of $C$ equals $10$. We draw 100 subsamples to calculate the instability defined in (\ref{eqn:ins}). Figure~\ref{fig:toy} shows the instability $\hat D(\lambda)$ over a range of values of $\log(\lambda)$ as well as the estimated rank $\hat{r}_{\lambda}$ based on the full data, from which we observe the following patterns.

\begin{itemize}
	\item[(i)] When $\lambda$ is small such that the estimated rank is larger than the true rank, the instability fluctuates but there is always a substantial gap from 0.
	\item[(ii)] In a sub-interval of the values of $\lambda$ where the estimated rank is equal to the true rank, the instability stays at zero.
	\item[(iii)] When $\lambda$ is large such that the estimated rank is less than the true rank, the instability fluctuates again and it can decrease to zero at certain values of $\lambda$.
\end{itemize}

The above observations motivate us to search from small to large for the first tuning parameter with which the instability is small enough. To achieve this objective, we consider a sequence of tuning parameters $\Lambda = \{\lambda_1, \ldots, \lambda_K\}$ where $\lambda_1 < \cdots < \lambda_K$. Based on this sequence of tuning parameters, we define the cumulative minimum instability for any $\lambda \in \Lambda$:
\begin{equation}\label{eqn:cmi}
\bar D(\lambda) =\min\{\hat D(\lambda'): \lambda' \in\Lambda,\  \lambda_{1}\leq \lambda' \leq\lambda\}.
\end{equation}
Then, we select the optimal tuning parameter as
\begin{equation*}
\hat\lambda = \min\{\lambda \in\Lambda: \bar D(\lambda) \le \eta\},
\end{equation*}
where $\eta$ is a small pre-specified threshold. In Section \ref{sec:the}, we will show that $\eta$ is interpretable as it is an explicit function of a key quantity in the theoretical property of StARS-RRR. Therefore, we are not simply replacing the problem of selecting $\lambda$ with the problem of selecting $\eta$. 

\iftrue
\begin{figure}
	\centering
	\includegraphics[width=\textwidth]{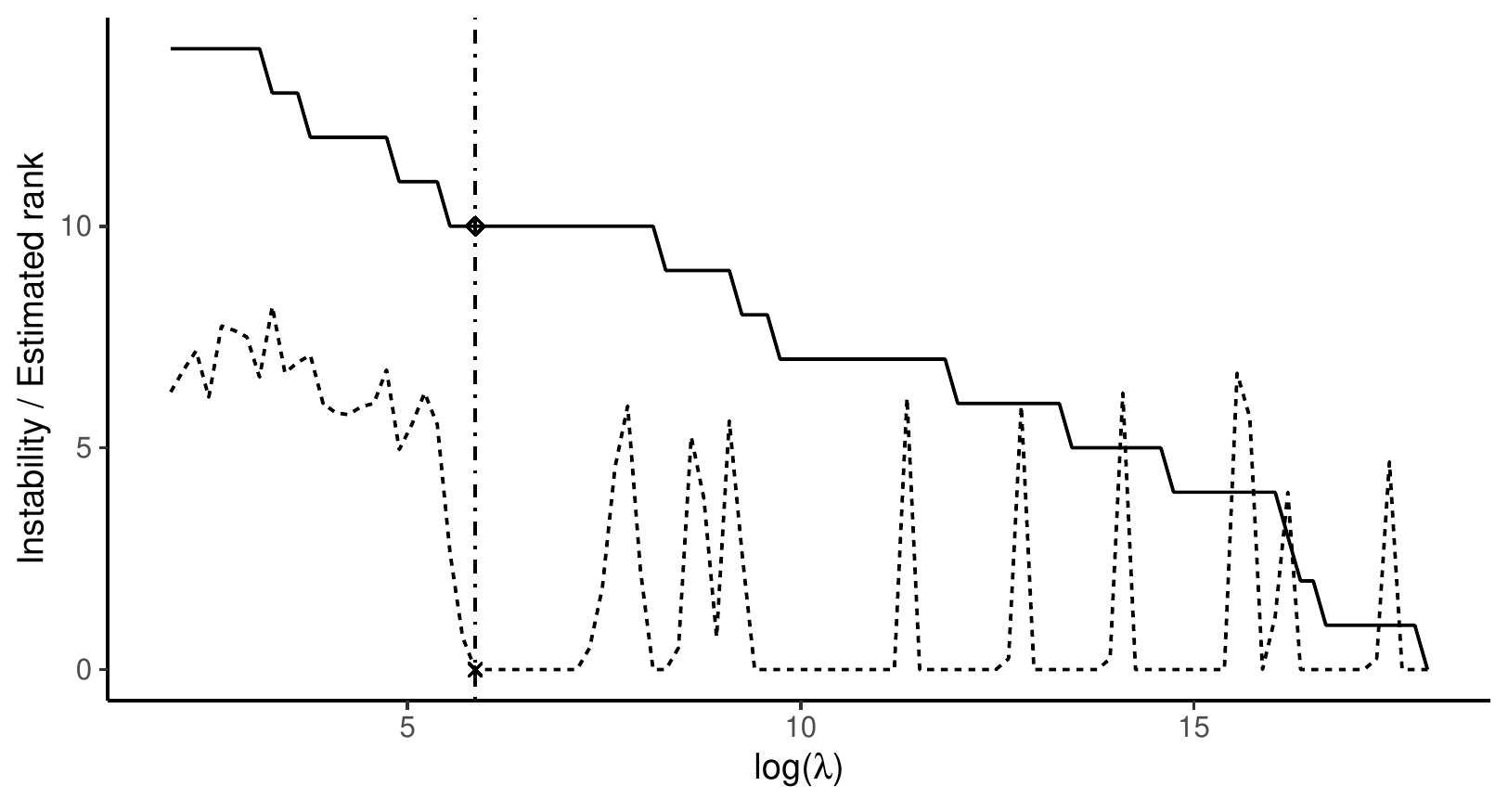}
	\caption{\label{fig:toy}The dotted line represents the instability relative to $\log(\lambda)$ and the solid line represents estimated rank from the full data relative to $\log(\lambda)$. In this simulated dataset, SNR is $ 2.047 $ and true rank is $ 10 $.}
\end{figure}
\fi

It is noteworthy that the selection rule in StARS-RRR is different from that in StARS \citep{liu2010stability}. When StARS was applied in graphical models, it starts with a large tuning parameter (or equivalently, an empty graph) and selects the first tuning parameter that achieves an instability lower than a threshold. However, StARS-RRR starts with a small tuning parameter (or equivalently, a high estimated rank) and selects the first tuning parameter that achieves an instability lower than a threshold. The reason why StARS-RRR searches the tuning parameter in a different direction is that, when the tuning parameter is too large so that the rank is underestimated, the estimated rank can still be stable across subsamples due to its discreteness. Therefore, if we start with a large tuning parameter and search for a stable estimated rank, we may select a tuning parameter that is too large and leads to an underestimated rank.

To demonstrate the performance of StARS-RRR, we indicate the selected tuning parameter in the previous simulated dataset as the vertical line in Figure~\ref{fig:toy}. It is seen that StARS-RRR selects $\lambda$ correctly such that the estimated rank is the same as the true rank. To summarize, we present the StARS-RRR method as in Algorithm \ref{alg1}.

\begin{algorithm}
	\caption{StARS-RRR}
	\label{alg1}
	\KwIn{Observations $ \{X_{i},Y_{i}\}_{i=1}^{n}$;
		Number of subsamples $N$;
		Size of subsample $b$;
		Threshold $\eta$;
		A list of candidate tuning parameters in an increasing order, $\Lambda = \{\lambda_{1},\ldots,\lambda_{K}\}$;}
	\KwOut{Optimal tuning parameter $\hat\lambda$;}
	Draw $N$ subsamples $S_j, j=1, \ldots, N$ of size $b$ from  $\{X_{i}, Y_{i}\}_{i=1}^{n}$ without replacement\;
	\For{each $k=1, \ldots,K$}{
		\For{each $j=1, \ldots, N$}{
			Apply the adaptive nuclear norm penalization method on data $S_j$ with tuning parameter $\lambda_k$ to obtain $\hat C_{\lambda_{k}}(S_j)$ with rank $\hat{r}_{\lambda_{k}}(S_j)$;
		}
		Compute the instability $\hat{D}(\lambda_{k})$ and then the cumulative minimum instability $ \bar{D}(\lambda_{k})$ by \eqref{eqn:ins} and \eqref{eqn:cmi}, respectively\;
		\If {$\bar{D}(\lambda_{k}) \le \eta$}{
			$\hat\lambda=\lambda_{k}$\;
			\textbf{break}
		}
	}
\end{algorithm}

Although our previous discussion is based on the adaptive nuclear norm penalization method for reduced-rank regression, StARS-RRR can be applied to any reduced-rank regression estimation methods that is not limited to this method. Furthermore, the proposed framework is not even limited to the reduced-rank regression model, it can be used for general matrix estimation problems as long as the objective is to determine the rank of a matrix.

\section{Rank Estimation Consistency}\label{sec:the}

In this section, we will show that the estimated rank from StARS-RRR is consistent to the true rank: the estimated rank is equal to the true rank with probability tending to one. In addition, we will show that there is an explicit relationship between the threshold $\eta$ employed in the StARS-RRR algorithm (Algorithm \ref{alg1}) and the probability that the rank is estimated correctly. Therefore, the threshold $\eta$ is an interpretable quantity and we are not simply replacing the problem of choosing $\lambda$ with the problem of choosing $\eta$. Different from the consistency result in the literature that requires a theoretical tuning parameter depending on unknown model parameters \citep{bunea2011optimal,chen2013reduced}, our result is practically more useful as it can be used to determine the tuning parameter in real applications.


Before presenting the consistency of rank estimation for StARS-RRR, we introduce some notation and impose necessary assumptions as follows. For any two real numbers $u$ and $v$, let $u\wedge v = \min(u,v)$ and $u\vee v = \max(u,v)$. Recall that $X$ and $Y$ are the design matrix and responses in the multivariate linear model in \eqref{eqn:mlm}. Further, let $X^b = (X_1,\ldots,X_b)^T$ and $Y^b = (Y_1,\ldots,Y_b)^T$ be the corresponding design matrix and responses from the subsample $\{X_i,Y_i\}_{i=1}^b$. Denote by $r_{\mathrm{x}}$ and  $r_{\mathrm{x}^b}$ the ranks of $X$ and $X^b$, respectively, and let $r^\ast$ be the true rank of the coefficient matrix $C$ in \eqref{eqn:mlm}. Finally, let $\hat{r}^b_{\lambda}$ be the estimated rank based on the subsample $\{X_i,Y_i\}_{i=1}^b$ using the adaptive nuclear norm penalization method \citep{chen2013reduced}. In addition, we impose the following assumptions.

\begin{assumption}
	\label{Assumption-1}
	The error matrix $E$ in (\ref{eqn:mlm}) has independent $N(0, \sigma^2)$ entries.
\end{assumption}
\begin{assumption}
	\label{Assumption-2}
	For a fixed $\theta > 0$, $d_{r^\ast}(X^b C) \ge 2 (1 + \theta) \sigma (\sqrt{r_{\mathrm{x}^b}} + \sqrt{q})$.
\end{assumption}
\begin{assumption}
	\label{Assumption-3}
	$r^{} \le r_{\mathrm{x}^b} \wedge q$ where $q$ is the number of responses in $Y$.
\end{assumption}

Assumptions \ref{Assumption-1} and \ref{Assumption-2} are almost identical to Assumptions 1 and 2 in \citet{chen2013reduced}, except that Assumption \ref{Assumption-2} is on the singular value of $X^b C$ for the subsample $X^b$ instead of the whole sample $X$. {\color{blue}Assumption \ref{Assumption-2} is a common type of assumption for subsampling based methods such as stability approaches. For example, Assumption (A2) in the original StARS paper \citep{liu2010stability} is also imposed on the subsamples of size $b < n$. It requires that from a subsample of size $b < n$ ``all estimated graphs using regularization parameters $\Lambda \ge \Lambda_0$ contain the true graph with high probability.''} Assumption \ref{Assumption-3} is a moderate assumption as long as $b$ is not too small because it is always true that $r^\ast \le r_{\mathrm{x}} \wedge q$. In practice, it is commonly assumed that $r^\ast$ is low for a reduced-rank regression problem.

As follows, we will establish the consistency of rank estimation for StARS-RRR in a few steps. First, we will show that the true variance of $\hat{r}^b_{\lambda}$ follows the patterns (i) and (ii) as observed from Figure~\ref{fig:toy}, i.e., (i) the true variance of $\hat{r}^b_{\lambda}$ stays away from 0 when $\lambda$ is small such that the rank is overestimated and (ii) the true variance of $\hat{r}^b_{\lambda}$ is very close to zero when the estimated rank is equal to the true rank. This result is summarized in Theorem \ref{Theorem-1}.

\begin{theorem}\label{Theorem-1}
	Suppose Assumptions \ref{Assumption-1}--\ref{Assumption-3} hold. For any $\delta \in [\exp\{-\theta^2 (\sqrt{r_{\mathrm{x}^b}} + \sqrt{q})^2 /8\}, 1/2)$ with a large enough $r_{\mathrm{x}^b} + q$, there exist $\lambda_l$, $\lambda_m$, $\lambda_h$ with $0 < \lambda_l \le \lambda_m \le \lambda_h \le [(1 + \theta/2) \sigma (\sqrt{r_{\mathrm{x}^b} } + \sqrt{q})]^{\gamma + 1}$ such that
	
	(a) when $\lambda_h \le \lambda \le [(1 + 3\theta/2) \sigma (\sqrt{r_{\mathrm{x}^b}} + \sqrt{q})]^{\gamma + 1}$, $P(\hat{r}^b_\lambda = r^\ast) \ge 1 - 2\exp\{-\theta^2 (\sqrt{r_{\mathrm{x}^b}} + \sqrt{q})^2 /8\}$ and $\mathrm{var}(\hat{r}^b_{\lambda}) \le 4 (r_{\mathrm{x}^b} \wedge q)^2 \exp\{-\theta^2 (\sqrt{r_{\mathrm{x}^b}} + \sqrt{q})^2 /8\}$, and further, $P(\hat{r}^b_\lambda = r^\ast) \to 1$ and $\mathrm{var}(\hat{r}^b_{\lambda}) \to 0$ as $r_{\mathrm{x}^b} + q \to \infty$;
	
	(b) when $\lambda \ge \lambda_m$, $P(\hat{r}^b_\lambda = r^\ast) \ge 1 - \delta - 2\exp\{-\theta^2 (\sqrt{r_{\mathrm{x}^b}} + \sqrt{q})^2 /8\}$;
	
	(c) when $0 < \lambda \le \lambda_m$, $P(\hat{r}^b_\lambda \ge r^\ast + 1) \ge \delta$, and when $\lambda_l \le \lambda \le \lambda_m$, $\mathrm{var}(\hat{r}^b_\lambda) \ge \delta (1-\delta)$.
\end{theorem}

On the one hand, part (a) of Theorem \ref{Theorem-1} provides the consistency of the estimated rank from a subsample of the data, which shows the adaptive nuclear norm penalization method is able to identify the correct rank with probability tending to one for an appropriate range of $\lambda$ values. This result is similar to Theorem 3 in \citet{chen2013reduced} with a slight distinction that we identify a range of $\lambda$ values for rank consistency instead of a single value of $\lambda$ as in \citet{chen2013reduced}. It is also obvious that the variance of the estimated rank tends to zero when the rank is correctly identified. On the other hand, the results in parts (b) and (c) provide additional information about the rank estimation when $\lambda$ is smaller. In part (b), when $\lambda \ge \lambda_m$, the adaptive nuclear norm penalization method achieves a slightly weaker result than that in part (a): the probability that the rank is correctly estimated is lowered by $\delta$ compared to part (a), although the estimated rank is still consistent as long as $\delta \to 0$. Part (c) discusses the case when $\lambda$ is even smaller, i.e., $\lambda_l \le \lambda \le \lambda_m$. There are two implications. With the probability at least $\delta$ the rank is overestimated and the variance of the estimated rank has a lower bound so that it stays away from zero. In summary, Theorem \ref{Theorem-1} shows the patterns (i) and (ii) observed in Figure \ref{fig:toy} theoretically for the true variance of the estimated rank from a subsample of the data.

Second, we will show that the sample variance of the estimated ranks from all the subsamples, i.e., the instability as defined in (\ref{eqn:ins}), is very close to the true variance of $\hat{r}^b_{\lambda}$. Hereby, we present this result in Theorem \ref{Theorem-2}.

\begin{theorem}\label{Theorem-2}
	For any $\lambda > 0$ such that $E(\hat{r}^b_\lambda) \ge 1/2$ and $t \in [6(r_{\mathrm{x}^b} \wedge q)^2/(N-1), 9(r_{\mathrm{x}} \wedge q)]$, 
	\begin{equation}
	P\left[ \left|S^{2}(\hat{r}^b_{\lambda}) - \mathrm{var}(\hat{r}^b_{\lambda})\right| > t \right] \le 6 \exp[ -n t^2/ \{162(r_{\mathrm{x}} \wedge q)^4 b\} ]. \label{eqn:concentration}
	\end{equation}
\end{theorem}

This result is similar to Theorem 1 in \citet{liu2010stability} although our focus is the difference between the sample variance and the true variance while \citet{liu2010stability} concerns the difference between the sample mean and the true mean. From (\ref{eqn:concentration}), it is seen that there is a trade-off between the difference $t$ and the probability on its right-hand side. For example, if $t$ is a fixed quantity, then the probability tends to zero as long as $(r_{\mathrm{x}} \wedge q)^4 b/n \to 0$. However, if one wishes to choose $t$ such that $t \to 0$ in order to achieve an asymptotically negligible difference, then the condition for the probability tending to zero becomes $(r_{\mathrm{x}} \wedge q)^4 b/(nt^2) \to 0$, depending on the convergence rate of $t$.

Combining Theorems \ref{Theorem-1} and \ref{Theorem-2}, we are able to validate the patterns (i) and (ii) as observed in Figure \ref{fig:toy} for the instability, which is presented in the following corollary.

\begin{corollary}\label{Corollary-1}
	Suppose Assumptions \ref{Assumption-1}--\ref{Assumption-3} hold and $\Lambda = \{\lambda_1,\ldots,\lambda_K\}$ is a grid of $K$ increasing positive values of $\lambda$. When $r_{\mathrm{x}^b} + q$ is large enough, with probability at least $1 - 6 K \exp[ -n \delta^2 / \{648 C^2 (r_{\mathrm{x}} \wedge q)^4 b \} ]$,
	\begin{align}
	\min \left\{ \hat{D}(\lambda) : \lambda \in \Lambda \cap [\lambda_l, \lambda_m] \right\} &\ge [(C-1)/C] \delta (1 - \delta), \label{hat.var.lower.bound} \\
	\max \left\{ \hat{D}(\lambda) : \lambda \in \Lambda \cap [\lambda_h, \{(1 + 3\theta/2) \sigma (\sqrt{r_{\mathrm{x}^b}} + \sqrt{q})\}^{\gamma + 1} ] \right\} &\le (2/C) \delta (1 - \delta), \label{hat.var.upper.bound}
	\end{align}
	for any fixed $C > 3$ and any $\delta \in [8 C (r_{\mathrm{x}^b} \wedge q)^2 \exp\{-\theta^2 (\sqrt{r_{\mathrm{x}^b}} + \sqrt{q})^2 /8\} \vee 12C (r_{\mathrm{x}^b} \wedge q)^2/(N-1), 1/2)$. 
\end{corollary}

In Corollary \ref{Corollary-1}, to ensure the interval for the possible values of $\delta$ is not empty, as $8 C (r_{\mathrm{x}^b} \wedge q)^2 \exp\{-\theta^2 (\sqrt{r_{\mathrm{x}^b}} + \sqrt{q})^2 /8\}$ is small with a large enough $r_{\mathrm{x}^b} + q$, one only needs to assume that $(r_{\mathrm{x}^b} \wedge q)^2/N$ is small. Moreover, for a large enough $C$, the lower bound in (\ref{hat.var.lower.bound}) is close to $\delta (1 - \delta)$ and the upper bound in (\ref{hat.var.upper.bound}) is close to 0. This verifies the patterns (i) and (ii) for the instability as observed in Figure \ref{fig:toy}.

Based on Corollary \ref{Corollary-1}, we can choose an appropriate threshold $\eta$ in StARS-RRR so that the optimal $\hat{\lambda}$ lies either in the interval $[\lambda_m, \lambda_h]$ or $[\lambda_h, \{(1 + 3\theta/2) \sigma (\sqrt{r_{\mathrm{x}^b}} + \sqrt{q})\}^{\gamma + 1}]$ {\color{blue}as long as the candidate tuning parameters $\Lambda = \{\lambda_1,\ldots,\lambda_K\}$ in StARS-RRR are all larger than $\lambda_l$}. Combining this result with parts (a) and (b) of Theorem \ref{Theorem-1} establishes the rank estimation consistency of StARS-RRR as presented in the following theorem.




\begin{theorem}\label{Theorem-3}
	Suppose Assumptions \ref{Assumption-1}--\ref{Assumption-3} hold and $\Lambda = \{\lambda_1,\ldots,\lambda_K\}$ is a grid of $K$ increasing positive values of $\lambda$ {\color{blue}such that $\lambda_1 \ge \lambda_l$ with $\lambda_l$ defined in Theorem \ref{Theorem-1}.} Let $\hat\lambda$ be the optimal tuning parameter selected by StARS-RRR described in Algorithm \ref{alg1} with a threshold $\eta = \delta(1-\delta)/2$ where $\delta \in [32 (r_{\mathrm{x}^b} \wedge q)^2 \exp\{-\theta^2 (\sqrt{r_{\mathrm{x}^b}} + \sqrt{q})^2 /8\} \vee 48 (r_{\mathrm{x}^b} \wedge q)^2/(N-1), 1/2)$, and let $\hat{r}^b_{\hat\lambda}$ be the estimated rank at $\hat\lambda$ from the subsample $X^b$ and $Y^b$. Then,
	\begin{equation}
	P(\hat{r}^b_{\hat\lambda} = r^\ast) \ge 1 - K \delta - 2 K \exp\{-\theta^2 (\sqrt{r_{\mathrm{x}^b}} + \sqrt{q})^2 /8\} -  6 K \exp[ -n \delta^2 / \{10368 (r_{\mathrm{x}} \wedge q)^4 b \}]. \label{consistency.probability}
	\end{equation}
	Furthermore, assume that $K$ is a fixed integer, that $r_{\mathrm{x}^b} + q \to \infty$, and that there exist $\alpha > 0$ and $\beta > 0$ such that $(r_{\mathrm{x}^b} \wedge q)^2 \exp\{-\theta^2 (\sqrt{r_{\mathrm{x}^b}} + \sqrt{q})^2/8\} = o(n^{-\alpha})$ and $(r_{\mathrm{x}} \wedge q)^4 b/n = o(n^{-\beta})$, as $n \to \infty$. Then, we can choose $\delta = n^{-(\alpha \wedge \beta)/2}$ in (\ref{consistency.probability}), which leads to
	\[
	P(\hat{r}^b_{\hat\lambda} = r^\ast) \to 1, \quad \text{as } n \to \infty.
	\]
\end{theorem}

Theorem \ref{Theorem-3} consists of two parts. First, it provides a finite sample lower bound for the probability with which the rank is correctly estimated when the tuning parameter is selected by StARS-RRR. It is interesting that there is an explicit relationship between this lower bound and the threshold $\eta$ used in StARS-RRR because $\eta = \delta(1-\delta)/2$ and the lower bound in (\ref{consistency.probability}) is also a known function of $\delta$. Therefore, this result gives $\eta$ an explicit interpretation by connecting it to the theoretical property of the estimated rank and makes the choice of the threshold meaningful in the StARS-RRR algorithm. Second, under further technical conditions, the estimated rank is asymptotically consistent to the true rank as the above-mentioned lower bound tends to one. While most of these conditions are common and also moderate, we note that the condition $(r_{\mathrm{x}} \wedge q)^4 b/n = o(n^{-\beta})$ imposes an upper bound on $r_{\mathrm{x}} \wedge q$ depending on $b$ and $n$. This condition arises from Theorem \ref{Theorem-2} and thus is similar to the condition that ensures the upper bound of probability in (\ref{eqn:concentration}) converges to zero.

{\color{blue} From the technical proofs of Theorems \ref{Theorem-1}--\ref{Theorem-3} (see the supplementary material), it is seen that our main theoretical result, the rank estimation consistency, is actually not limited to the adaptive nuclear norm penalization method. In fact, any rank estimation method that results in an estimated rank as in (\ref{equ:est-r}) will enjoy the results in Theorem \ref{Theorem-3}. For example, the $L_0$ penalized estimator in \citet{bunea2011optimal} also results in a similar form of the estimated rank:
\begin{equation*}
\hat{r}_{\lambda}= \max\{r: d_r(PY) > \lambda\},
\end{equation*}
with $\lambda$ being the tuning parameter for the $L_0$ penalty. Therefore, the rank estimation consistency would also hold for the $L_0$ penalized estimator in \citet{bunea2011optimal} as long as we replace $\gamma$ in (\ref{equ:est-r}) by $0$.}

Compared to the partial sparsistency property as established for StARS in \citet{liu2010stability}, our result is apparently stronger. The partial sparsistency result shows that with probability tending to one the true edges of a graph belong to the estimated edge set using the optimal tuning parameter selected by StARS on a subsample of size $b$, which is only a ``one-direction'' result. By contrast, our result is a ``two-direction'' result that establishes the rank estimation consistency for StARS-RRR.

In practice, after the optimal tuning parameter is selected by StARS-RRR, one often performs the reduced-rank regression on the full data with this optimal tuning parameter. Thus, we recommend choosing $b$ sufficiently large so that the behavior of the subsample $X^b$ and $Y^b$ is similar to that of the full data. On the other hand, $b$ cannot be too large as implied by the technical condition $(r_{\mathrm{x}} \wedge q)^4 b/n = o(n^{-\beta})$. Therefore, an appropriate size $b$ needs to be selected under a particular setting of $X$, $q$, and $n$. This philosophy is similar to how $b$ is chosen in StARS although our technical condition about $b$ is slightly more complicated than that for StARS. {\color{blue}In our numerical studies in Section \ref{sec:num}, we set $\eta = 0.001$, $b = 0.7n$, and $N = 100$ in Algorithm \ref{alg1} for StARS-RRR.}

\section{Numerical Experiments}\label{sec:num}

\subsection{Simulation}\label{sec:sim}

In this subsection, we compare the finite sample performance of the rank determination via StARS-RRR and other approaches including AIC, BIC, GIC, BICP, GCV and cross validation (CV) on simulated data. 

We adopt the same simulation settings from \citet{bunea2011optimal}. Specifically, the coefficient matrix $ C $ is generated by $ C=sC_{1}C_{2}^T $, where $ s>0 $, $ C_{1}\in\mathbb{R}^{p\times r^{\ast}} $, $ C_{2}\in\mathbb{R}^{q\times r^{\ast}} $. All entries in $ C_{1} $ and $ C_{2} $ are drawn randomly from { ${N}(0,1)$}. The design matrix $ X $ is generated by $ X=X_{0}\Gamma^{1/2}$, where $ X_{0}=X_{1}X_{2}^T$, $ X_{1}\in\mathbb{R}^{n\times r_{\mathrm{x}}} $, $ X_{2}\in\mathbb{R}^{p\times r_{\mathrm{x}}} $, and $ \Gamma=(\Gamma_{ij})_{p\times p} $ with $ \Gamma_{ij}=\rho^{\lvert i-j\rvert} $, $i,j=1,\ldots,p$. The response matrix $ Y $ is then generated by $ Y=XC+E $, where the elements of $ E $ are independent random variables from { ${N}(0,1)$}. Thus, the simulation model is characterized by the sample size $ n $, the number of predictors $ p $, the number of responses $ q $, the rank of the design matrix $r_{\mathrm{x}} $, the true model rank $ r^{\ast} $, the correlation coefficient between the adjacent predictors $\rho$, and the signal $s$.

We will explore two different model settings, where Model I is a low-dimensional case  with $(n, p, q, r_{\mathrm{x}}, r^\ast) = (500, 25, 25, 15, 10)$ and Model II is a high-dimensional case with $(n, p, q, r_{\mathrm{x}}, r^\ast) = (80, 100, 100, 30, 8)$. In Model I, $p$ and $q$ are relatively small compared with $n$; while in Model II, $p$ and $q$ are relatively large. Although these finite-sample settings may not align perfectly with the technical conditions in our large-sample theory in Section \ref{sec:the}, they still deserve numerical investigations to see how StARS-RRR performs in practice. In addition, for each model, we set $\rho$ to be 0.1, 0.5, or 0.9, which stands for weak, moderate, or strong dependence between the predictors. We consider six different values of $s$ such that the signal-to-noise ratio (SNR) ranges from 1 to 3 roundly. Since the $ r^{\ast}$th largest singular value of $ XC $, i.e., $ d_{r^{\ast}}(XC) $, measures the signal strength and the largest singular value of the projected noise matrix $ \mathrm{PE} =X(X^TX)^{-}X^TE $, i.e., $ d_{ 1 }(\mathrm{PE}) $, measures the noise level, the SNR is defined as $ d_{r^{\ast}}(XC)/d_{ 1 }(\mathrm{PE}) $ \citep{chen2013reduced}.

We apply the reduced-rank regression via the adaptive nuclear norm penalization to fit the full data with the optimal tuning parameter selected by AIC, BIC, GIC, BICP, GCV, CV (5-fold), and StARS-RRR. In StARS-RRR, we set the threshold $\eta$ as 0.0001. For each method, a total of 500 simulation replications are conducted. 
To compare the performance of the aforementioned methods, we consider several performance measures. The first measure is the rank recovery ratio, defined as the proportion of $\{\hat{r}=r^{\ast}\} $ over all replications, where $\hat{r}$ is the estimated rank. The method with a higher rank recovery ratio is more effective in rank determination. The second and third measures are the rank overestimate ratio and the rank underestimate ratio, which are defined by the proportion of $\{\hat{r}>r^{\ast}\} $ and $\{\hat{r}<r^{\ast}\} $ over all replications, respectively. The fourth measure is the bias of the estimated rank, defined as the mean difference between the estimated rank and the true rank. A better performing method should have a lower bias in terms of magnitude.
%

Table \ref{tab:ratios} summarizes the rank recovery, underestimate, and overestimate ratios. For model I, it is clear that StARS-RRR outperforms the other methods in terms of every performance measure when SNR is moderate to high. In particular, the rank recovery ratio of StARS-RRR is at least $10\%$ higher than those of information criterion based methods. When the SNR is low, i.e., $\text{SNR}<1.5$, GCV has the best performance, followed by AIC, CV, and StARS-RRR. In this case, the StARS-RRR estimators from the subsamples might be dominated by the noise, and this leads to underestimation of the rank. In contrast, GCV, CV, and AIC tend to overestimate the rank and result in a more complicate model, which would explain why these methods have good performance when SNR is very low.

For Model II, different from Model I where the sample size is sufficient compared to the dimension, BICP and GIC can rarely recover the true rank. This is probably because that the penalty in BICP and GIC contains the term $\log(pq)$ (Table \ref{tab:loss}), which becomes excessive when the sample size is small and the dimension is high. As a result, the estimated ranks from these two methods are much smaller than the true rank. This can also be observed from Table~\ref{tab:bias}, where the mean biases of the estimated ranks from BICP and GIC are negative. {\color{blue}It is clear that StARS-RRR has a much better performance than BICP and GIC in recovering the true rank among a wide range of SNR.} Regardless of whether the correlation between the predictors is low, medium, or high, StARS-RRR can always achieve a rank recovery ratio of at least $97\%$ when the SNR is greater than $1.4$, a commonly occurred case in practice. StARS-RRR also slightly outperforms CV although the performance of CV in Model II is better than in Model I. In summary, when the sample size is small and the dimension is high, StARS-RRR can obtain a reliable result in determining the rank in a reduced-rank regression model.

\begin{table}
	\caption{\label{tab:ratios}Rank recovery (left), underestimate (middle), and overestimate (right) ratios (in percentage) in the simulation study.}
	\centering
	\resizebox{\textwidth}{!}{
	\begin{tabular}{ccccccccc}
		\hline
		s& SNR & AIC & BIC & GIC & BICP & GCV &  CV & StARS-RRR \\
		\hline
		\hline
		\multicolumn{9}{c}{Model \uppercase\expandafter{\romannumeral1}, $\rho = 0.1$} \\
		30 & 1.07 & (81,4,15) & (17,83,0) & (2,98,0) & (4,96,0) & (85,4,12) & (80,5,15) & (63,37,0) \\ 
  45 & 1.6 & (86,1,14) & (62,38,0) & (36,64,0) & (42,58,0) & (87,1,12) & (86,1,13) & (92,8,0) \\ 
  52 & 1.85 & (85,0,14) & (76,24,0) & (53,47,0) & (59,41,0) & (87,0,12) & (87,0,13) & (95,4,0) \\ 
  60 & 2.14 & (86,0,14) & (87,13,0) & (68,32,0) & (74,26,0) & (87,0,13) & (87,0,13) & (98,2,0) \\ 
  70 & 2.49 & (86,0,14) & (94,6,0) & (83,17,0) & (86,14,0) & (87,0,13) & (87,0,13) & (98,1,0) \\ 
  85 & 3.03 & (85,0,15) & (98,2,0) & (93,7,0) & (95,5,0) & (87,0,13) & (87,0,13) & (99,0,0) \\ 
		
		\hline
		\multicolumn{9}{c}{Model \uppercase\expandafter{\romannumeral1}, $\rho = 0.5$} \\
		35 & 1.1 & (83,4,13) & (23,77,0) & (4,96,0) & (6,94,0) & (84,5,11) & (82,4,14) & (63,37,0) \\ 
  40 & 1.26 & (85,2,13) & (36,64,0) & (12,88,0) & (18,82,0) & (86,3,11) & (83,3,14) & (73,26,1) \\ 
  50 & 1.57 & (86,1,13) & (60,40,0) & (35,65,0) & (42,58,0) & (87,1,12) & (85,1,13) & (89,10,1) \\ 
  70 & 2.2 & (86,0,14) & (86,14,0) & (68,32,0) & (73,27,0) & (88,0,12) & (87,0,13) & (97,3,1) \\ 
  80 & 2.52 & (86,0,14) & (93,7,0) & (80,20,0) & (84,16,0) & (88,0,12) & (87,0,13) & (99,1,1) \\ 
  95 & 2.99 & (86,0,14) & (97,3,0) & (90,10,0) & (92,8,0) & (88,0,12) & (86,0,14) & (99,1,1) \\ 
		\hline
		\multicolumn{9}{c}{Model \uppercase\expandafter{\romannumeral1}, $\rho = 0.9$} \\
		70 & 1.08 & (78,5,17) & (20,80,0) & (4,96,0) & (8,92,0) & (80,5,14) & (78,6,16) & (63,37,0) \\ 
  80 & 1.24 & (82,3,16) & (33,67,0) & (11,89,0) & (16,84,0) & (84,3,13) & (82,3,15) & (73,27,0) \\ 
  100 & 1.55 & (83,1,16) & (56,44,0) & (32,68,0) & (39,61,0) & (86,1,14) & (84,1,15) & (86,14,0) \\ 
  135 & 2.09 & (85,0,15) & (83,17,0) & (64,36,0) & (70,30,0) & (88,0,12) & (86,0,14) & (97,3,0) \\ 
  175 & 2.71 & (85,0,15) & (94,6,0) & (84,16,0) & (88,12,0) & (87,0,13) & (87,0,13) & (99,1,0) \\ 
  200 & 3.1 & (85,0,15) & (97,3,0) & (91,9,0) & (93,7,0) & (87,0,13) & (86,0,14) & (99,0,0) \\
	
		\hline\hline
		\multicolumn{9}{c}{Model \uppercase\expandafter{\romannumeral2}, $\rho = 0.1$} \\
		8 & 1.16 & (43,0,57) & (0,100,0) & (0,100,0) & (0,100,0) & (86,0,14) & (90,1,9) & (77,23,0) \\ 
  10 & 1.45 & (48,0,52) & (12,88,0) & (0,100,0) & (0,100,0) & (91,0,9) & (97,0,3) & (98,2,0) \\ 
  12 & 1.73 & (50,0,50) & (54,46,0) & (0,100,0) & (0,100,0) & (93,0,7) & (99,0,1) & (100,0,0) \\ 
  14 & 2.02 & (50,0,50) & (83,17,0) & (0,100,0) & (0,100,0) & (93,0,7) & (99,0,1) & (100,0,0) \\ 
  18 & 2.6 & (47,0,53) & (99,1,0) & (5,95,0) & (24,76,0) & (93,0,7) & (99,0,1) & (100,0,0) \\ 
  22 & 3.18 & (45,0,55) & (100,0,0) & (65,35,0) & (80,20,0) & (92,0,8) & (99,0,1) & (100,0,0) \\
  
		\hline
		\multicolumn{9}{c}{Model \uppercase\expandafter{\romannumeral2}, $\rho = 0.5$} \\
		8 & 1.12 & (42,0,58) & (0,100,0) & (0,100,0) & (0,100,0) & (86,0,14) & (89,0,11) & (69,30,0) \\ 
  10 & 1.4 & (48,0,52) & (5,95,0) & (0,100,0) & (0,100,0) & (89,0,11) & (95,0,5) & (95,4,0) \\ 
  12 & 1.68 & (51,0,49) & (45,55,0) & (0,100,0) & (0,100,0) & (92,0,8) & (98,0,2) & (100,0,0) \\ 
  15 & 2.09 & (52,0,48) & (88,12,0) & (0,100,0) & (0,100,0) & (94,0,6) & (99,0,1) & (100,0,0) \\ 
  18 & 2.51 & (52,0,48) & (99,1,0) & (2,98,0) & (17,83,0) & (93,0,7) & (99,0,1) & (100,0,0) \\ 
  22 & 3.07 & (48,0,52) & (100,0,0) & (55,45,0) & (75,25,0) & (91,0,9) & (100,0,0) & (100,0,0) \\ 
		\hline
		\multicolumn{9}{c}{Model \uppercase\expandafter{\romannumeral2}, $\rho = 0.9$} \\
  11 & 1.05 & (41,0,59) & (0,100,0) & (0,100,0) & (0,100,0) & (81,0,19) & (83,1,16) & (58,42,0) \\ 
  13 & 1.24 & (44,0,56) & (3,97,0) & (0,100,0) & (0,100,0) & (85,0,15) & (90,0,10) & (86,14,0) \\ 
  16 & 1.52 & (47,0,53) & (26,74,0) & (0,100,0) & (0,100,0) & (89,0,11) & (96,0,4) & (99,1,0) \\ 
  21 & 2 & (49,0,51) & (82,18,0) & (0,100,0) & (1,99,0) & (91,0,9) & (99,0,1) & (100,0,0) \\ 
  26 & 2.48 & (49,0,51) & (98,2,0) & (15,85,0) & (25,75,0) & (90,0,10) & (99,0,1) & (100,0,0) \\ 
  32 & 3.05 & (47,0,53) & (100,0,0) & (61,39,0) & (74,26,0) & (90,0,10) & (99,0,1) & (100,0,0) \\ 
		
		\hline
		\hline
	\end{tabular}}
\end{table}

\begin{table}
	\caption{\label{tab:bias}Mean and standard error (in parenthesis) of bias of the estimated rank in the simulation study.}
	\centering
	\resizebox{\textwidth}{!}{
		\begin{tabular}{cccccccc}
			\hline\hline
			 SNR & AIC & BIC & GIC & BICP & GCV & CV  & StARS-RRR\\
			\hline
			\multicolumn{8}{c}{Model \uppercase\expandafter{\romannumeral1}, $\rho = 0.1$} \\
		1.07 & 0.12 (0.45) & -1.11 (0.68) & -1.73 (0.71) & -1.55 (0.7) & 0.09 (0.41) & 0.11 (0.46) & -0.4 (0.65) \\ 
  1.6 & 0.14 (0.39) & -0.4 (0.53) & -0.76 (0.65) & -0.65 (0.6) & 0.12 (0.37) & 0.13 (0.38) & -0.06 (0.37) \\ 
  1.85 & 0.15 (0.39) & -0.24 (0.43) & -0.51 (0.57) & -0.43 (0.54) & 0.13 (0.38) & 0.13 (0.38) & -0.03 (0.32) \\ 
  2.14 & 0.15 (0.39) & -0.13 (0.34) & -0.32 (0.47) & -0.26 (0.45) & 0.13 (0.36) & 0.14 (0.37) & 0 (0.28) \\ 
  2.49 & 0.15 (0.39) & -0.06 (0.23) & -0.17 (0.38) & -0.14 (0.35) & 0.14 (0.37) & 0.13 (0.36) & 0 (0.28) \\ 
  3.03 & 0.16 (0.39) & -0.02 (0.13) & -0.07 (0.26) & -0.05 (0.21) & 0.13 (0.36) & 0.14 (0.37) & 0 (0.19) \\ 
			\hline
			\multicolumn{8}{c}{Model \uppercase\expandafter{\romannumeral1}, $\rho = 0.5$} \\
			1.1 & 0.09 (0.43) & -0.98 (0.67) & -1.57 (0.69) & -1.42 (0.68) & 0.06 (0.42) & 0.11 (0.44) & -0.39 (0.63) \\ 
  1.26 & 0.12 (0.4) & -0.74 (0.63) & -1.22 (0.69) & -1.08 (0.68) & 0.09 (0.4) & 0.12 (0.43) & -0.25 (0.58) \\ 
  1.57 & 0.13 (0.39) & -0.42 (0.53) & -0.75 (0.62) & -0.66 (0.62) & 0.11 (0.37) & 0.13 (0.42) & -0.07 (0.48) \\ 
  2.2 & 0.15 (0.41) & -0.14 (0.35) & -0.33 (0.48) & -0.27 (0.45) & 0.12 (0.36) & 0.14 (0.4) & 0 (0.35) \\ 
  2.52 & 0.15 (0.38) & -0.07 (0.26) & -0.2 (0.4) & -0.16 (0.36) & 0.13 (0.37) & 0.14 (0.39) & 0.02 (0.32) \\ 
  2.99 & 0.15 (0.38) & -0.03 (0.18) & -0.1 (0.29) & -0.08 (0.27) & 0.13 (0.37) & 0.15 (0.4) & 0.02 (0.32) \\ 
			\hline
			\multicolumn{8}{c}{Model \uppercase\expandafter{\romannumeral1}, $\rho = 0.9$} \\
			1.08 & 0.13 (0.49) & -1 (0.64) & -1.55 (0.7) & -1.39 (0.69) & 0.1 (0.47) & 0.11 (0.52) & -0.42 (0.58) \\ 
  1.24 & 0.15 (0.45) & -0.77 (0.63) & -1.23 (0.67) & -1.09 (0.65) & 0.12 (0.43) & 0.14 (0.48) & -0.29 (0.49) \\ 
  1.55 & 0.17 (0.44) & -0.46 (0.54) & -0.79 (0.62) & -0.69 (0.61) & 0.14 (0.41) & 0.16 (0.46) & -0.14 (0.36) \\ 
  2.09 & 0.17 (0.42) & -0.17 (0.37) & -0.37 (0.51) & -0.31 (0.48) & 0.13 (0.38) & 0.17 (0.46) & -0.02 (0.24) \\ 
  2.71 & 0.16 (0.41) & -0.06 (0.23) & -0.16 (0.37) & -0.12 (0.32) & 0.14 (0.39) & 0.15 (0.43) & 0 (0.2) \\ 
  3.1 & 0.17 (0.42) & -0.03 (0.16) & -0.09 (0.29) & -0.07 (0.25) & 0.14 (0.39) & 0.16 (0.45) & 0 (0.19) \\ 
  
			\hline\hline
			\multicolumn{8}{c}{Model \uppercase\expandafter{\romannumeral2}, $\rho = 0.1$} \\
			1.16 & 1.16 (1.33) & -3.05 (0.85) & -7.63 (0.52) & -7.39 (0.62) & 0.16 (0.41) & 0.09 (0.33) & -0.79 (1.86) \\ 
  1.45 & 1.11 (1.37) & -1.45 (0.83) & -7.1 (0.73) & -6.65 (0.82) & 0.11 (0.36) & 0.03 (0.21) & -0.03 (0.2) \\ 
  1.73 & 1.16 (1.47) & -0.52 (0.61) & -6.52 (0.87) & -5.76 (1.03) & 0.09 (0.34) & 0.02 (0.14) & 0 (0.06) \\ 
  2.02 & 1.21 (1.55) & -0.17 (0.4) & -5.78 (1.15) & -4.61 (1.4) & 0.08 (0.33) & 0.01 (0.12) & 0 (0) \\ 
  2.6 & 1.37 (1.66) & -0.01 (0.09) & -3.24 (1.94) & -1.45 (1.27) & 0.09 (0.35) & 0.01 (0.11) & 0 (0) \\ 
  3.18 & 1.42 (1.67) & 0 (0) & -0.46 (0.75) & -0.22 (0.46) & 0.11 (0.41) & 0.01 (0.15) & 0 (0) \\ 
			\hline
			\multicolumn{8}{c}{Model \uppercase\expandafter{\romannumeral2}, $\rho = 0.5$} \\
			
			1.12 & 1.19 (1.36) & -3.19 (0.84) & -7.54 (0.55) & -7.31 (0.64) & 0.15 (0.41) & 0.11 (0.34) & -0.98 (2.27) \\ 
  1.4 & 1.1 (1.37) & -1.66 (0.83) & -6.97 (0.76) & -6.46 (0.86) & 0.13 (0.38) & 0.06 (0.25) & -0.04 (1.08) \\ 
  1.68 & 1.1 (1.42) & -0.66 (0.68) & -6.34 (0.94) & -5.67 (1.06) & 0.1 (0.34) & 0.02 (0.17) & 0.04 (0.94) \\ 
  2.09 & 1.15 (1.5) & -0.13 (0.35) & -5.2 (1.32) & -3.95 (1.38) & 0.08 (0.33) & 0.01 (0.11) & 0.04 (0.94) \\ 
  2.51 & 1.19 (1.55) & -0.01 (0.11) & -3.4 (1.74) & -1.71 (1.32) & 0.09 (0.36) & 0.01 (0.09) & 0.04 (0.94) \\ 
  3.07 & 1.3 (1.59) & 0 (0) & -0.64 (0.86) & -0.29 (0.54) & 0.11 (0.37) & 0 (0.09) & 0.04 (0.94) \\ 
  
			\hline
			\multicolumn{8}{c}{Model \uppercase\expandafter{\romannumeral2}, $\rho = 0.9$} \\
			1.05 & 1.26 (1.4) & -2.6 (0.73) & -5.92 (0.75) & -5.58 (0.8) & 0.21 (0.48) & 0.16 (0.42) & -0.75 (1.6) \\ 
  1.24 & 1.2 (1.4) & -1.72 (0.73) & -5.38 (0.84) & -4.9 (0.87) & 0.18 (0.45) & 0.11 (0.35) & -0.15 (1.1) \\ 
  1.52 & 1.19 (1.48) & -0.91 (0.66) & -4.47 (0.97) & -3.8 (0.89) & 0.13 (0.4) & 0.05 (0.23) & 0.04 (0.94) \\ 
  2 & 1.23 (1.57) & -0.19 (0.39) & -2.88 (0.98) & -2.2 (0.91) & 0.11 (0.36) & 0.02 (0.14) & 0.04 (0.94) \\ 
  2.48 & 1.3 (1.64) & -0.02 (0.15) & -1.38 (0.88) & -1 (0.75) & 0.12 (0.38) & 0.01 (0.09) & 0.04 (0.94) \\ 
  3.05 & 1.37 (1.67) & 0 (0) & -0.46 (0.62) & -0.28 (0.51) & 0.12 (0.38) & 0.01 (0.1) & 0.04 (0.94) \\ 
			\hline
			\hline
	\end{tabular}}
\end{table}

Table \ref{tab:bias} tabulates the mean and standard error of the bias of the estimated rank over simulation replicates. From this table, we observe that the biases of the estimated ranks from AIC, GCV, and CV are always greater than 0, for both Model I and Model II, which indicates that these three methods tend to yield a more complex model. Furthermore, their biases do not quite vanish when SNR increases in Model I and/or Model II, which means that these methods tend to overestimate the rank of the coefficient matrix regardless of the signal-to-noise ratio. By contrast, BIC, BICP, and GIC tend to underestimate the rank for both Model I and Model II, as evidenced by the negative signs of their mean biases. However, the magnitude of the bias decreases when SNR increases, which implies that these three methods perform better with a higher SNR. Finally, StARS-RRR achieves the smallest or the second smallest magnitude of bias in most settings except for very low SNR's. Overall, StARS-RRR performs the best among all the methods we investigated in rank determination for reduced-rank regression models.

{\color{blue}Based on referee's suggestions, we also conduct additional simulation studies to evaluate the prediction performance of StARS-RRR, as well as sensitivity analyses for the hyperparameters $\eta$, $N$, and $b$ in StARS-RRR. For more details, please refer to Section B of the supplementary material.}

\subsection{Application to breast cancer data}

In this subsection, we apply StARS-RRR to a real dataset to show its effectiveness in rank determination. In particular, we consider the breast cancer data \citep{witten2009a}, consisting of the gene expression measurements and comparative genomic hybridization (CGH) measurements for $n = 89$ patients. This dataset has been studied in previous work \citep{bunea2011optimal,chen2013reduced} and is available in the R package \textit{PMA} \citep{witten2009a}. The question of interest is to investigate the relationship between the DNA copy-number variations and gene expression profiles for the patients. We will use reduced-rank regression to model the copy number variations based on the gene expression profiles \citep{geng2011virtual,zhou2012prediction}. A reduced-rank regression model yields a low-rank coefficient matrix, with the estimated rank representing the number of linear combinations of gene expression measurements that enter into the prediction of CGH measurements. These linear combinations of gene expression measurements can be regarded as biologically functional pathways that affect the DNA copy number variations.

For the purpose of illustration, we analyzed the data on chromosome 13, where there are $ p=319 $ gene expression measurements and $ q=58 $ CGH measurements. The adaptive nuclear norm penalization method was used to estimate the coefficient matrix, and StARS-RRR was applied to determine the optimal rank, together with the other approaches used in the simulation. The estimated ranks are presented in the top panel of Table~\ref{app1}. On the one hand, AIC and GCV estimate the rank as 57 and 26, respectively, which seem to overestimate the number of functional pathways of practical interest. On the other hand, BIC, GIC, BICP and CV estimate the rank as 1 or 2 and haven't revealed enough biological relationships for further investigation. Instead, StARS-RRR reveals three linear combinations of gene expressions that potentially affect copy number variations, which include a reasonable number of biological pathways that deserve further investigation.

To visualize the relationship between the DNA copy-number variations and gene expression profiles, we also plot the estimated coefficient matrix with the tuning parameter selected by StARS-RRR in the form of a heat map in Figure \ref{figapp}. It is visually clear that there are three sets of CGH measurements, each of which follows a similar relationship with the gene profiles. The left 24 CGH measurements have a strong relationship with the genes as the coefficients have the largest magnitude among the three sets and have both positive and negative signs. The middle 3 CGH measurements have a moderate relationship that is similar to the above measurements in terms of signs although the magnitude of the coefficients is much smaller. The right 31 CGH measurements have a weak relationship with the genes as their coefficients are small in magnitude.

\iftrue
\begin{figure}
	\centering
	\includegraphics[width=\textwidth]{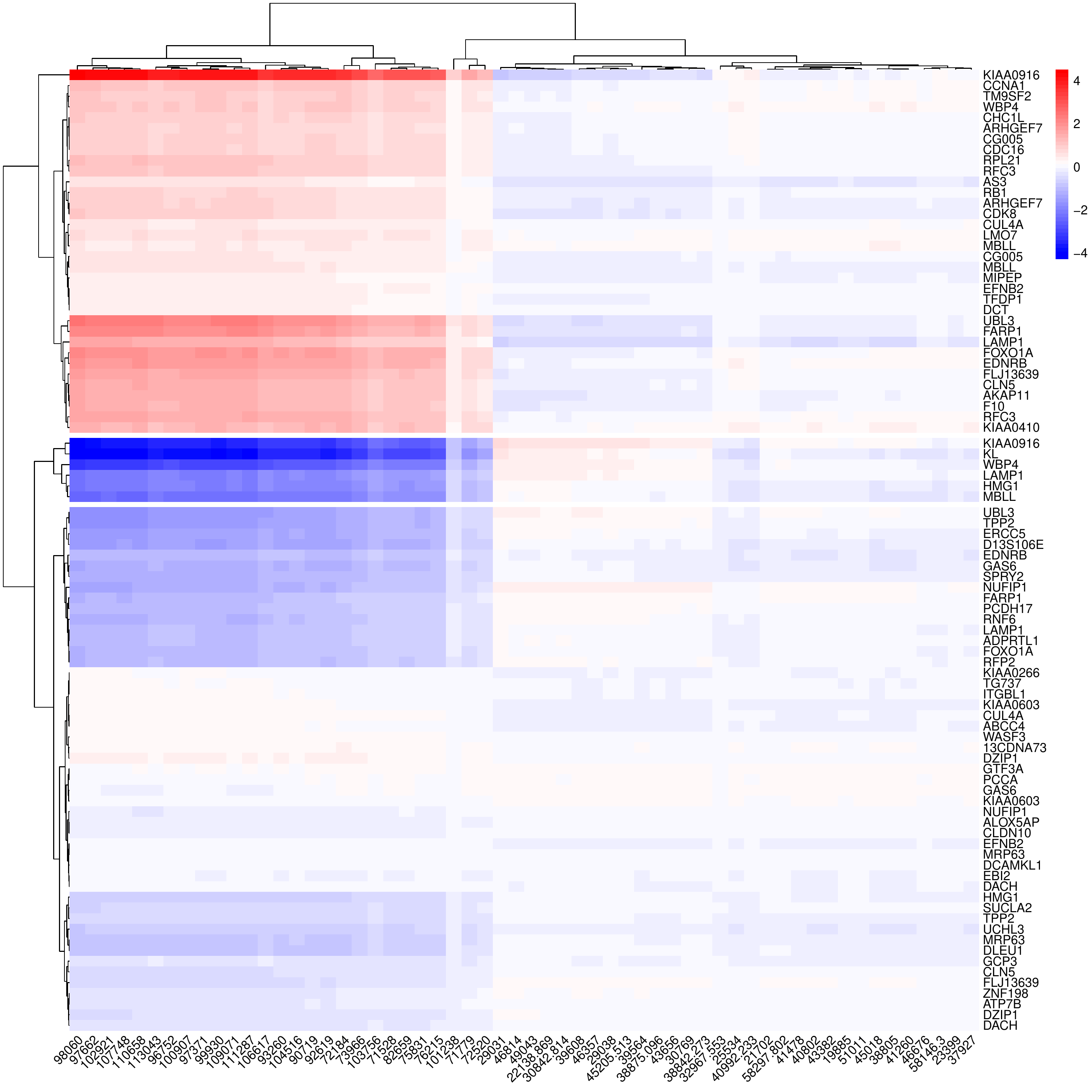}
	\caption{\label{figapp}The heat map of coefficient matrix obtained by the StARS-RRR for each CGH spot (row) and the gene (column). Genes with all of its  coefficients being 0 are not shown.}
\end{figure}
\fi

To provide further insight into the performance of different methods, we carried out the following random-splitting process for 100 times. The data were randomly split into a training set of size $ n_{\text{train}}=79 $ and a test set of size $ n_{\text{test}}=10 $. We first estimate the rank using the aforementioned methods, and then refit the model using a ridge generalization of reduced-rank regression model \citep{izenman2008modern} to derive the final estimated coefficient matrix $\hat{C}$. Finally, we calculate the mean squared prediction error (MSPE) as
$\mathrm{MSPE} = 100\times \lVert  Y_{\text{test}}-X_{\text{test}}\hat{C}\rVert_{\mathrm{F}}^{2}/(qn_{\text{test}})$,
where $X_{\text{test}}$ and $Y_{\text{test}}$ are the predictors and responses in the test set.

\begin{table}
	\caption{\label{app1}Comparison of the model fits to the real data for various tuning parameter selection methods. The mean squared prediction errors (MSPE) and the estimated ranks (Rank) are reported, with their standard errors in the parentheses.}
	\centering
	\resizebox{\textwidth}{!}{
		\begin{tabular}{lccccccc}
			\hline
			& AIC & BIC & GIC & BICP & GCV & CV & StARS-RRR \\
			\hline
			\multicolumn{8}{c}{Full Data} \\
			Rank & 57 & 2 & 1 & 2 & 26 & {\color{blue}0} & 3\\
			\hline	
			\multicolumn{8}{c}{Random-Splitting Process}\\	
  Rank & {\color{blue}56.98 (0.14)} & {\color{blue}49.87 (18.51)} & {\color{blue}1.00 (0.00)} & {\color{blue}1.06 (0.24)} & {\color{blue}30.76 (12.42)} & {\color{blue}0.80 (0.40)} & {\color{blue}2.65 (0.52)} \\ 
  MSPE & {\color{blue}4.32 (0.98)} & {\color{blue}4.20 (1.02)} & {\color{blue}3.71 (1.28)} & {\color{blue}3.71 (1.27)} & {\color{blue}4.25 (0.98)} & {\color{blue}3.84 (1.37)} & {\color{blue}3.41 (0.95)} \\ 
			\hline
	\end{tabular}}
\end{table}

Similar to the full-data performance, the bottom panel of Table~\ref{app1} shows that the tuning parameter selection methods can be divided into three groups according to their estimated ranks. The first group consists of AIC, BIC, and GCV, whose estimated ranks are quite large. The coefficient matrix has $58$ columns, so the coefficient matrices estimated by AIC and BIC are of almost full rank. Therefore, these methods result in a complex model that might overfit the data. The second group includes GIC, BICP, and CV, where the average estimated rank is close to 1. This might imply an underestimation of the rank. Underestimation of the rank leads to the lack of information to be extracted, which also explains why they underperform in terms of prediction accuracy. The last group is StARS-RRR, whose average estimated rank is between 2 and 3, a more reasonable rank than the other methods. Moreover, the MSPE of StARS-RRR is lower than all the other methods, which further convinces us that StARS-RRR yields an accurate estimation of rank for this dataset.

\section{Discussion}\label{sec:dis}

In this article, we propose a new method based on the stability approach to select the tuning parameter for reduced-rank regression. To the best of our knowledge, it is the first time that the stability approach is used in this framework. Our main contribution is twofold. First, we set up a general framework of the stability approach for rank determination including a new definition of instability and a new tuning parameter selection rule based on the instability. This framework is generally applicable to any matrix estimation problem and is referred to as StARS-RRR when specifically applied to reduced-rank regression. Second, we show that the rank determined by StARS-RRR is consistent to the true rank when the adaptive nuclear norm penalization is used. In fact, we provide a finite sample lower bound of the probability with which the rank is estimated correctly. Interestingly, there is an explicit relationship between this lower bound and the threshold $\eta$ used to select the tuning parameter in StARS-RRR, giving the method high interpretability.

Although StARS-RRR performs satisfactorily in both simulated and real data, it still has a few {\color{blue}limitations that need attention and/or }deserve further studies. First, the current definition of instability emphasizes the stability of the estimated rank but there could be alternative definitions. For example, if one concerns more the stability of the subspace corresponding to the estimated coefficient matrix instead of its rank, the instability could be defined as the variation of such subspaces estimated from the randomly draw subsamples. We refer to \citet{taeb2020false} for recent development of subspace stability in low-rank matrix estimation. {\color{blue} Second, we established the rank estimation consistency for StARS-RRR, which makes it more theoretically sound than most information criterion methods when applied to reduced-rank regression. However, we must point out that most information criterion methods do not need any assumptions on the design matrix, while StARS-RRR assumes an independent and identically distributed sample. This limitation needs to be taken into consideration in the application of StARS-RRR to dependent or heterogeneous data.} 
Third, this article represents only the first effort to apply the stability approach to low-rank matrix estimation and thus there are still many unsolved questions to explore. For example, if both sparsity and low rank are desired as in the sparse reduced-rank regression, it will be an interesting topic to extend StARS-RRR so that variable selection consistency and rank estimation consistency can be achieved simultaneously. Generalization of StARS-RRR to other matrix estimation problems such as low-rank matrix completion also warrants future investigation.

\section*{Acknowledgement}
Wen's research was supported in part by National Science Foundation of China under Grants 12171449 and 11801540; and Natural Science Foundation of Anhui Province under Grant BJ2040170017. Jiang's research was supported in part by the U.S. National Institutes of Health under Grant R01GM126549.

\section*{Supplementary Material}

Supplementary material includes technical proofs of the theoretical results as well as additional numerical results.

\bibliographystyle{apalike}
\bibliography{reference}

\end{document}


\baselineskip=20pt
\renewcommand{\thesection}{\Alph{section}}

\begin{center}
\textbf{\large Supplemental Material for ``Stability Approach to Regularization Selection for Reduced Rank Regression"}
\end{center}

This document provides detailed proofs for the theorems described in the main text, as well as additional numerical results as recommended by the referees.

\section{Technical Proofs}
\renewcommand{\theequation}{A.\arabic{equation}}
\setcounter{equation}{0}

\subsection{Lemmas}

\begin{lemma}\label{Lemma-1}
	Suppose Assumption~1 holds. Then for any $t > 0$, $E\{d_1(P^b E^b)\} \le \sigma (\sqrt{r_{\mathrm{x}^b}} + \sqrt{q})$ and $P [d_1(P^b E^b) \ge E\{d_1(P^b E^b)\} + \sigma t] \le \exp(-t^2/2)$.
\end{lemma}

\begin{proof}
	Lemma \ref{Lemma-1} is identical to Lemma 3 in \citet{bunea2011optimal}.
\end{proof}

\begin{lemma}\label{Lemma-2}
	Suppose Assumptions~1--2 hold. For any $\delta$ satisfying that $\exp\{-\theta^2 (\sqrt{r_{\mathrm{x}^b}} + \sqrt{q})^2 /8\} \le \delta < 1/2$ with a large enough $r_{\mathrm{x}^b} + q$, there exist $\lambda_l,\ \lambda_m,\ \lambda_h$ with $0 < \lambda_l \le \lambda_m \le \lambda_h \le [(1 + \theta/2) \sigma (\sqrt{r_{\mathrm{x}^b}} + \sqrt{q})]^{\gamma + 1}$ such that $P(d_{r^{\ast} + 1}(P^b Y^b) > \lambda_l^{1/(\gamma + 1)}) = 1 - \delta$, $P(d_{r^{\ast} + 1}(P^b Y^b) > \lambda_m^{1/(\gamma + 1)}) = \delta$, and $P(d_{r^{\ast} + 1}(P^b Y^b) > \lambda_h^{1/(\gamma + 1)}) = \exp\{-\theta^2 (\sqrt{r_{\mathrm{x}^b}} + \sqrt{q})^2 /8\}$.
	
\end{lemma}

\begin{proof}
	
	On the one hand, $d_{r^{\ast} + 1}(P^b Y^b)$ is a continuous function of $P^b Y^b$. Further, $P^b Y^b = X^b C + P^b E^b$ where $E^b$ has independent $N(0, \sigma^2)$ entries, then $d_{r^{\ast} + 1}(P^b Y^b)$ has a continuous distribution. Denote by $q_{\delta}$ and $q_{1-\delta}$ the $\delta$-quantile and $(1-\delta)$-quantile of such a distribution, respectively. Let $\lambda_l = q_{\delta}^{\gamma + 1}$ and $\lambda_m = q_{1-\delta}^{\gamma + 1}$, then
	\[P[d_{r^{\ast} + 1}(P^b Y^b) > \lambda_l^{1/(\gamma + 1)}] = 1-\delta \quad \text{and} \quad P[d_{r^{\ast} + 1}(P^b Y^b) > \lambda_m^{1/(\gamma + 1)}] = \delta.\]
	
	On the other hand, based on Weyl’s inequalities on singular values \citep{franklin2012matrix} and observing that $P^b Y^b = X^b C + P^b E^b$, we have 
	\begin{equation}\label{weyl}
	|d_s(P^b Y^b) - d_s(X^b C)| \le d_1(P^b E^b), \quad s = 1, \ldots, b \wedge q.
	\end{equation}
	When $s = r^{\ast} + 1$, $d_{r^{\ast} + 1}(P^b Y^b) \le d_{r^{\ast} + 1}(X^b C) + d_1(P^b E^b)$ and by Lemma \ref{Lemma-1},
	\[P[d_{r^{\ast} + 1}(P^b Y^b) < d_{r^{\ast} + 1}(X^b C) + \sigma (\sqrt{r_{\mathrm{x}^b}} + \sqrt{q}) + \sigma t] \ge 1 - \exp(-t^2/2),\]
	for any $t > 0$.
	Since $d_{r^{\ast} + 1}(X^b C) = 0$, let $t = \theta (\sqrt{r_{\mathrm{x}^b}} + \sqrt{q}) /2$, we have 
	\[ P[d_{r^{\ast} + 1}(P^b Y^b) > (1 + \theta/2) \sigma (\sqrt{r_{\mathrm{x}^b}} + \sqrt{q}) ] \le \exp\{-\theta^2 (\sqrt{r_{\mathrm{x}^b}} + \sqrt{q})^2 /8\}.\]
	Therefore, there exists $\lambda_h \le [(1 + \theta/2) \sigma (\sqrt{r_{\mathrm{x}}} + \sqrt{q})]^{\gamma + 1}$ such that 
	\begin{equation}\label{dr*+1}
	P[d_{r^{\ast} + 1}(P^b Y^b) > \lambda_h^{1/(\gamma + 1)}] = \exp\{-\theta^2 (\sqrt{r_{\mathrm{x}^b}} + \sqrt{q})^2 /8\}.
	\end{equation}
	
	Finally, $\lambda_l \le \lambda_m$ is implied by the fact that $\delta < 1/2$. In addition, $\lambda_m \le \lambda_h$ is implied by the fact that $\delta \ge \exp\{-\theta^2 (\sqrt{r_{\mathrm{x}^b}} + \sqrt{q})^2 /8\}$.
	
	
\end{proof}

\subsection{Proof of Theorem 3.1}

For part (a), by the Weyl's inequalities, $d_{r^{\ast}}(P^b Y^b) \ge d_{r^{\ast}}(X^b C) - d_1(P^b E^b)$ and by Lemma \ref{Lemma-1}, 
\[P[d_{r^{\ast}}(P^b Y^b) > d_{r^{\ast}}(X^b C) - \sigma (\sqrt{r_{\mathrm{x}^b}} + \sqrt{q}) - \sigma t] \ge 1 - \exp(-t^2/2),\]
for any $t > 0$. Similar to the proof of Lemma \ref{Lemma-2}, since $d_{r^{\ast}}(X^b C) \ge 2 (1 + \theta) \sigma (\sqrt{r_{\mathrm{x}^b}} + \sqrt{q}) $,
\[ P[d_{r^{\ast}}(P^b Y^b) > (1 + 2\theta) \sigma (\sqrt{r_{\mathrm{x}^b}} + \sqrt{q}) - \sigma t] \ge 1 - \exp(-t^2/2). \]
Letting $t = \theta (\sqrt{r_{\mathrm{x}^b}} + \sqrt{q}) /2$, we have that
\[ P[d_{r^{\ast}}(P^b Y^b) > (1 + 3\theta/2) \sigma (\sqrt{r_{\mathrm{x}^b}} + \sqrt{q}) ] \ge 1 - \exp\{-\theta^2 (\sqrt{r_{\mathrm{x}^b}} + \sqrt{q})^2 /8\}. \]
When $\lambda \le [(1 + 3\theta/2) \sigma (\sqrt{r_{\mathrm{x}^b}} + \sqrt{q})]^{\gamma + 1}$, this leads to 
\begin{equation}\label{dr*}
P[d_{r^{\ast}}(P^b Y^b) > \lambda^{1/(\gamma + 1)}] \ge 1 - \exp\{-\theta^2 (\sqrt{r_{\mathrm{x}^b}} + \sqrt{q})^2 /8\}.
\end{equation}
Since $\lambda_h \le \lambda \le [(1 + 3\theta/2) \sigma (\sqrt{r_{\mathrm{x}^b}} + \sqrt{q})]^{\gamma + 1}$, combining (\ref{dr*+1}) and (\ref{dr*}) leads to 
\[ P[d_{r^{\ast}+1}(P^b Y^b) \le \lambda^{1/(\gamma + 1)} < d_{r^{\ast}}(P^b Y^b)] \ge 1 - 2\exp\{-\theta^2 (\sqrt{r_{\mathrm{x}^b}} + \sqrt{q})^2 /8\}. \]
Since $\hat{r}^b_{\lambda}=\max\{r:d_{r}(P^b Y^b)>\lambda^{1/(\gamma + 1)}\}$ as in (4),
\[ P(\hat{r}^b_{\lambda} = r^{\ast}) \ge 1 - 2\exp\{-\theta^2 (\sqrt{r_{\mathrm{x}^b}} + \sqrt{q})^2 /8\}.\]
Furthermore, as $P^b Y^b$ has at most $r_{\mathrm{x}^b} \wedge q$ positive singular values,
\begin{align*}
&\mathrm{var}(\hat{r}^b_{\lambda}) \\= {} & E(\hat{r}^b_{\lambda}) - [E(\hat{r}^b_{\lambda})]^2 \\
\le {} & (r^{\ast})^2 [1 - 2\exp\{-\theta^2 (\sqrt{r_{\mathrm{x}^b}} + \sqrt{q})^2 /8\}] + 2 (r_{\mathrm{x}^b} \wedge q)^2 \exp\{-\theta^2 (\sqrt{r_{\mathrm{x}^b}} + \sqrt{q})^2 /8\} \\
& - (r^{\ast})^2 [1 - 2\exp\{-\theta^2 (\sqrt{r_{\mathrm{x}^b}} + \sqrt{q})^2 /8\}]^2 \\
= {} & 2 (r^{\ast})^2 \exp\{-\theta^2 (\sqrt{r_{\mathrm{x}^b}} + \sqrt{q})^2 /8\} [1 - 2\exp\{-\theta^2 (\sqrt{r_{\mathrm{x}^b}} + \sqrt{q})^2 /8\}] \\
&+ 2 (r_{\mathrm{x}^b} \wedge q)^2 \exp\{-\theta^2 (\sqrt{r_{\mathrm{x}^b}} + \sqrt{q})^2 /8\}\\
\le {} & 4 (r_{\mathrm{x}^b} \wedge q)^2 \exp\{-\theta^2 (\sqrt{r_{\mathrm{x}^b}} + \sqrt{q})^2 /8\},
\end{align*} 
because $r^{\ast} \le r_{\mathrm{x}^b} \wedge q$ as in Assumption 3.

For part (b), as $\hat{r}^b_{\lambda}$ is a decreasing function of $\lambda$, $P(\hat{r}^b_\lambda \ge r^{\ast}) \ge 1 - 2\exp\{-\theta^2 (\sqrt{r_{\mathrm{x}^b}} + \sqrt{q})^2 /8\}$ is implied by the result in part (a). Further, for any $\lambda$ satisfying that $\lambda \ge \lambda_m$,
\[P[d_{r^{\ast} + 1}(P^b Y^b) > \lambda^{1/(\gamma + 1)}] \le \delta, \]
which implies that $P(\hat{r}^b_\lambda \ge r^{\ast} + 1) \le \delta$ from the definition of $\hat{r}^b_{\lambda}$. Thus, 
\[P(\hat{r}^b_\lambda = r^{\ast}) = P(\hat{r}^b_\lambda \ge r^{\ast}) - P(\hat{r}^b_\lambda \ge r^{\ast} + 1) \ge 1 - \delta - 2\exp\{-\theta^2 (\sqrt{r_{\mathrm{x}^b}} + \sqrt{q})^2 /8\}.\]

For part (c), from Lemma \ref{Lemma-2}, $P(d_{r^{\ast} + 1}(P^b Y^b) > \lambda_m^{1/(\gamma + 1)}) = \delta$. Then, for any $\lambda$ satisfying $0 \le \lambda \le \lambda_m$, $P(d_{r^{\ast} + 1}(P^b Y^b) > \lambda^{1/(\gamma + 1)}) \ge \delta$, which leads to $P(\hat{r}^b_{\lambda} \ge r^{\ast} + 1) \ge \delta$ from the definition of $\hat{r}^b_{\lambda}$.

The estimated rank $\hat{r}^b_{\lambda}$ can be alternatively written as $\hat{r}^b_{\lambda} = \sum_{i=1}^{r_{\mathrm{x}^b} \wedge q} I\{d_{i}(P^b Y^b)>\lambda^{1/(\gamma + 1)}\}$, where $I(\cdot)$ is the indicator function. Therefore,
\begin{align*}
\mathrm{var}(\hat{r}_{\lambda}^{b}) = {} & \sum_{i=1}^{r_{\mathrm{x}^b} \wedge q} \mathrm{var}[ I\{d_{i}(P^b Y^b)>\lambda^{1/(\gamma + 1)}\}] \\
&+ \sum_{i=1}^{r_{\mathrm{x}^b} \wedge q} \sum_{j=1}^{r_{\mathrm{x}^b} \wedge q} \mathrm{cov}[ I\{d_{i}(P^b Y^b)>\lambda^{1/(\gamma + 1)}\}, I\{d_{j}(P^b Y^b)>\lambda^{1/(\gamma + 1)}\}].
\end{align*}
It is noteworthy that, for $1 \le i, j \le r_{\mathrm{x}^b} \wedge q$,
\begin{align*}
& \mathrm{cov}[ I\{d_{i}(P^b Y^b) > \lambda^{1/(\gamma + 1)}\}, I\{d_{j}(P^b Y^b) > \lambda^{1/(\gamma + 1)}\} ] \\
= {} & P[ d_{i}(P^b Y^b) > \lambda^{1/(\gamma + 1)}, d_{j}(P^b Y^b) > \lambda^{1/(\gamma + 1)} ] \\
&- P[ d_{i}(P^b Y^b) > \lambda^{1/(\gamma + 1)} ] P[ d_{j}(P^b Y^b) > \lambda^{1/(\gamma + 1)} ] \\
= {} & P[ d_{i \wedge j}(P^b Y^b) > \lambda^{1/(\gamma + 1)} ] [1 - P\{ d_{i \vee j}(P^b Y^b) > \lambda^{1/(\gamma + 1)} \}] \\
\ge {} & 0.
\end{align*}
Therefore, 
\begin{equation*} 
\mathrm{var}(\hat{r}^b_{\lambda}) \ge \sum_{i=1}^{r_{\mathrm{x}^b} \wedge q} \mathrm{var}[ I\{d_{i}(P^b Y^b)>\lambda^{1/(\gamma + 1)}\}] \ge \mathrm{var}[ I\{d_{r^{\ast} + 1}(P^b Y^b)>\lambda^{1/(\gamma + 1)}\}].
\end{equation*}
When $\lambda_l \le \lambda \le \lambda_m$, $\delta \le P[d_{r^{\ast} + 1}(P^b Y^b)>\lambda^{1/(\gamma + 1)}] \le 1-\delta$, which leads to that $\mathrm{var}(\hat{r}^b_{\lambda}) \ge \delta (1-\delta)$.

\subsection{Proof of Theorem 3.2}


For any $\lambda > 0$, denote $\mu_2^b(\lambda) = \mathrm{var}(\hat{r}^b_{\lambda})$ as the variance of the estimated rank from the reduced rank regression based on the subsample $(Y_1,X_1),\ldots,(Y_b,X_b)$. The parameter $\mu_2^b(\lambda)$ can be written as $\mu_2^b(\lambda) = m_2^b(\lambda) - \{m_1^b(\lambda)\}^2$, where $m_2^b(\lambda) = \mathrm{E}[(\hat{r}^b_{\lambda})^2]$ and $m_1^b(\lambda) = \mathrm{E}(\hat{r}^b_{\lambda})$. Note that $\hat{m}_1^b(\lambda) = \sum_{i=1}^N \hat{r}_{\lambda}(S_i)/N$ and $\hat{m}_2^b(\lambda) = \sum_{i=1}^N \hat{r}^2_{\lambda}(S_i)/N$ are their corresponding U-statistics of order $b$, and further that $0 \le \hat{r}_{\lambda}(S_i) \le r_{\mathrm{x}} \wedge q$ for $i=1,\ldots,N$. Hence, by Hoeffding's inequality for U-statistics \citep{serfling2009approximation}, we have, for any $t > 0$,
\begin{align}
P\left[\left| \hat{m}_1^b(\lambda) - m_1^b(\lambda) \right| > t\right] & \le 2 \exp[-2 n t^2/\{(r_{\mathrm{x}} \wedge q)^2 b\}], \label{hoeffding.1}\\
P\left[\left| \hat{m}_2^b(\lambda) - m_2^b(\lambda) \right| > t\right] & \le 2 \exp[-2 n t^2/\{(r_{\mathrm{x}} \wedge q)^4 b\}]. \label{hoeffding.2}
\end{align}
In addition, $m_1^b(\lambda) \le r_{\mathrm{x}^b} \wedge q$ and $m_2^b(\lambda) \le (r_{\mathrm{x}^b} \wedge q)^2$ for any $\lambda > 0$ as $P^b Y^b$ has at most $r_{\mathrm{x}^b} \wedge q$ positive singular values. Setting $t = m_1^b(\lambda) $ in (\ref{hoeffding.1}) leads to that 
\begin{equation*}
P\left[\hat{m}_1^b(\lambda) \le 2 m_1^b(\lambda) \right] \ge 1 - 2 \exp[-2 n \{m_1^b(\lambda)\}^2/\{(r_{\mathrm{x}} \wedge q)^2 b\}].
\end{equation*}
For any $\lambda > 0$ such that $m_1^b(\lambda) \ge 1/2$, with probability at least $1 - 2 \exp[- n /\{2(r_{\mathrm{x}} \wedge q)^2 b\}]$, $\hat{m}_1^b(\lambda) \le 2(r_{\mathrm{x}^b} \wedge q)$. 

From the definition of $S^{2}(\hat{r}_{\lambda})$ in (5), with probability at least $1 - 2 \exp[- n /\{2(r_{\mathrm{x}} \wedge q)^2 b\}]$,
\begin{align*}
& \left|S^{2}(\hat{r}_{\lambda}) - \mu_2^b(\lambda) \right| \\ \le {}& \frac{N}{N-1} \left| \hat{m}_2^b(\lambda) - m_2^b(\lambda) \right| +\frac{N}{N-1} \left| \left\{ \hat{m}_1^b(\lambda) \right\}^2 - \{m_1^b(\lambda)\}^2 \right|  + \frac{1}{N-1} \left[m_2^b(\lambda) + \{m_1^b(\lambda)\}^2\right] \notag \\
\le {}& \frac{N}{N-1} \left| \hat{m}_2^b(\lambda) - m_2^b(\lambda) \right| + \frac{N}{N-1} \left| \{ \hat{m}_1^b(\lambda) + m_1^b(\lambda) \} \{ \hat{m}_1^b(\lambda) - m_1^b(\lambda) \} \right|  + \frac{2}{N-1} (r_{\mathrm{x}^b} \wedge q)^2 \notag\\
\le {}&\frac{N}{N-1} \left| \hat{m}_2^b(\lambda) - m_2^b(\lambda) \right| + \frac{N}{N-1} \left| 3(r_{\mathrm{x}^b} \wedge q) \{ \hat{m}_1^b(\lambda) - m_1^b(\lambda) \} \right| + \frac{2}{N-1} (r_{\mathrm{x}^b} \wedge q)^2.
\end{align*}

For $t \ge 6(r_{\mathrm{x}^b} \wedge q)^2/(N-1) $, 
\[\frac{2}{N-1} (r_{\mathrm{x}^b} \wedge q)^2 \le \frac{t}{3}.\]
Thus, for any $\lambda > 0$ such that $m_1^b(\lambda) \ge 1/2$ and $6(r_{\mathrm{x}^b} \wedge q)^2/(N-1) \le t \le 9 (r_{\mathrm{x}} \wedge q)$,
\begin{align*}
& P\left[\left| S^{2}(\hat{r}_{\lambda}) - \mu_2^b(\lambda) \right| > t \right] \\
\le{}& 2 \exp[- n /\{2(r_{\mathrm{x}} \wedge q)^2 b\}] + P\left[ \frac{N}{N-1} \left| \hat{m}_2^b(\lambda) - m_2^b(\lambda) \right| > \frac{t}{3} \right] \\ &+ P\left[ \frac{N}{N-1} \left|3(r_{\mathrm{x}^b} \wedge q) \{ \hat{m}_1^b(\lambda) - m_1^b(\lambda) \} \right| > \frac{t}{3} \right]\\
\le {}& 2 \exp[- n /\{2(r_{\mathrm{x}} \wedge q)^2 b\}] + P\left[ | \hat{m}_2^b(\lambda) - m_2^b(\lambda) | > \frac{t}{6} \right] \\ &+ P\left[ | \hat{m}_1^b(\lambda) - m_1^b(\lambda) | > \frac{t}{18(r_{\mathrm{x}^b} \wedge q)} \right] \\
\le {}& 2 \exp[- n /\{2(r_{\mathrm{x}} \wedge q)^2 b\}] + 2 \exp[- n t^2/\{(18(r_{\mathrm{x}} \wedge q)^4 b\}] \\& + 2 \exp[ -n t^2/ \{162(r_{\mathrm{x}^b} \wedge q)^2 (r_{\mathrm{x}} \wedge q)^2 b\} ] \\
\le {}& 6 \exp[ -n t^2/ \{162(r_{\mathrm{x}} \wedge q)^4 b\} ].
\end{align*}

\subsection{Proof of Corollary 3.1}

Denote $\lambda_u = [(1 + 3\theta/2) \sigma (\sqrt{r_{\mathrm{x}^b}} + \sqrt{q})]^{\gamma + 1}$ for simplicity of notation. Based on Theorem 3.1, for any $\delta \in [\exp\{-\theta^2 (\sqrt{r_{\mathrm{x}^b}} + \sqrt{q})^2 /8\}, 1/2)$
\begin{align}
\mu_2^b(\lambda) \ge \delta (1-\delta), & \quad \text{when}\ \lambda_l \le \lambda \le \lambda_m, \label{var.lower.bound} \\
\mu_2^b(\lambda)\le 4 (r_{\mathrm{x}^b} \wedge q)^2 \exp\{-\theta^2 (\sqrt{r_{\mathrm{x}^b}} + \sqrt{q})^2 /8\}, & \quad \text{when}\ \lambda_h \le \lambda \le \lambda_u. \label{var.upper.bound.1} 
\end{align}
Thus, for any fixed $C > 3$, if $\delta \in [8 C (r_{\mathrm{x}^b} \wedge q)^2 \exp\{-\theta^2 (\sqrt{r_{\mathrm{x}^b}} + \sqrt{q})^2 /8\}, 1/2)$ with a large enough $r_{\mathrm{x}^b} + q$, (\ref{var.lower.bound}) still holds and (\ref{var.upper.bound.1}) becomes
\begin{equation}
\mu_2^b(\lambda)\le \delta(1-\delta)/C, \quad \text{when}\ \lambda_h \le \lambda \le \lambda_u. \label{var.upper.bound.2} 
\end{equation}

Based on Theorem 3.2, for any $t$ satisfying that $6(r_{\mathrm{x}^b} \wedge q)^2/(N-1) \le t \le 9 (r_{\mathrm{x}} \wedge q)$,
\begin{align*}
& P\left[ \max_{\lambda \in \Lambda: m_1^b(\lambda) \ge 1/2} \left|S^{2}(\hat{r}_{\lambda}) - \mu_2^b(\lambda)\right| > t \right] 
\\ \le {} & 6 \sum_{\lambda \in \Lambda: m_1^b(\lambda) \ge 1/2} \exp[ -n t^2/ \{162(r_{\mathrm{x}} \wedge q)^4 b\} ] \\
\le {} & 6 K \exp[ -n t^2 / \{162(r_{\mathrm{x}} \wedge q)^4 b\} ].
\end{align*}
Therefore, when $12C (r_{\mathrm{x}^b} \wedge q)^2/(N-1) \le \delta < 1/2$, we can set $t = (1/C) \delta (1-\delta)$ in the above inequality, which leads to
\begin{align*}
&  P\left[ \max_{\lambda \in \Lambda: m_1^b(\lambda) \ge 1/2} \left|S^{2}(\hat{r}_{\lambda}) - \mu_2^b(\lambda)\right| > (1/C) \delta (1-\delta) \right] \\ \le {} & 6 K \exp[ -n \delta^2 (1-\delta)^2 / \{162 C^2 (r_{\mathrm{x}} \wedge q)^4 b \} ] \\
\le {} & 6 K \exp[ -n \delta^2 / \{648 C^2 (r_{\mathrm{x}} \wedge q)^4 b \} ].
\end{align*}
In other words, when $\max[8 C (r_{\mathrm{x}^b} \wedge q)^2 \exp\{-\theta^2 (\sqrt{r_{\mathrm{x}^b}} + \sqrt{q})^2 /8\}, 12C (r_{\mathrm{x}^b} \wedge q)^2/(N-1)] \le \delta < 1/2$, with probability at least $1 - 6 K \exp[ -n \delta^2 / \{648 C^2 (r_{\mathrm{x}} \wedge q)^4 b \} ]$,
\begin{equation}
\max_{\lambda \in \Lambda: m_1^b(\lambda) \ge 1/2} \left|S^{2}(\hat{r}_{\lambda}) - \mu_2^b(\lambda)\right| \le (1/C) \delta (1-\delta).  \label{var.diff}
\end{equation}

Whenever $\lambda \in [\lambda_l, \lambda_m]$ or $\lambda \in [\lambda_h, \lambda_u]$, $m_1^b(\lambda) = E(\hat{r}^b_\lambda) \ge 1/2$ with a large enough $r_{\mathrm{x}} + q$ based on Theorem 1. Combining (\ref{var.lower.bound}), (\ref{var.upper.bound.2}), and (\ref{var.diff}), 
\begin{align*}
\min \left\{ S^{2}(\hat{r}_{\lambda}) : \lambda \in \Lambda \cap [\lambda_l, \lambda_m] \right\} &\ge [(C-1)/C] \delta (1 - \delta), \\
\max \left\{ S^{2}(\hat{r}_{\lambda}) : \lambda \in \Lambda \cap [\lambda_h, \lambda_u] \right\} &\le (2/C) \delta (1 - \delta).
\end{align*}


\subsection{Proof of Theorem 3.3}

In Corollary 3.1, setting $C = 4$ leads to that, for any $\delta$ satisfying $\max[32 (r_{\mathrm{x}^b} \wedge q)^2 \exp\{-\theta^2 (\sqrt{r_{\mathrm{x}^b}} + \sqrt{q})^2 /8\}, 48 (r_{\mathrm{x}^b} \wedge q)^2/(N-1)] \le \delta < 1/2$, the following bounds
\begin{align*}
\min \left\{ \hat{D}(\lambda) : \lambda \in \Lambda \cap [\lambda_l, \lambda_m] \right\} & \ge 3\delta(1-\delta)/4 > \eta, \\
\max \left\{ \hat{D}(\lambda) : \lambda \in \Lambda \cap [\lambda_h, \lambda_u] \right\} &\le \delta(1-\delta)/2 = \eta,
\end{align*}
hold with probability at least $1 - 6 K \exp[ -n \delta^2 / \{10368 (r_{\mathrm{x}} \wedge q)^4 b \}]$, which implies that the optimal tuning parameter selected by StARS-RRR $\hat{\lambda}$ must lie in $[\lambda_m, \lambda_h]$ or $[\lambda_h, \lambda_u]$, i.e., 
\[ P(\hat\lambda \in [\lambda_m, \lambda_u]) \ge 1 - 6 K \exp[ -n \delta^2 / \{10368 (r_{\mathrm{x}} \wedge q)^4 b \}]. \]

We can bound the probability of $\{\hat{r}^b_{\hat\lambda} \ne r^{\ast}\}$ as follows:
\begin{align*}
& P(\hat{r}^b_{\hat\lambda} \ne r^{\ast})\\ \le{} & P(\hat{r}^b_{\hat\lambda} \ne r^{\ast}, \hat\lambda \in \Lambda \cap [\lambda_m, \lambda_h]) + P(\hat{r}^b_{\hat\lambda} \ne r^{\ast}, \hat\lambda \in \Lambda \cap [\lambda_h, \lambda_u]) + P(\hat\lambda \notin [\lambda_m, \lambda_u])
\\ \le{}& \sum_{\lambda \in \Lambda \cap [\lambda_m, \lambda_h]} P(\hat{r}^b_{\lambda} \ne r^{\ast}, \hat\lambda = \lambda) + \sum_{\lambda \in \Lambda \cap [\lambda_h, \lambda_u]} P(\hat{r}^b_{\lambda} \ne r^{\ast}, \hat\lambda = \lambda) + P(\hat\lambda \notin [\lambda_m, \lambda_u]).
\end{align*}
From parts (a) and (b) in Theorem 3.1, 
\begin{align*}
P(\hat{r}^b_{\lambda} \ne r^{\ast}) \le \delta + 2\exp\{-\theta^2 (\sqrt{r_{\mathrm{x}^b}} + \sqrt{q})^2 /8\}, \quad &\text{when } \lambda \in  [\lambda_m, \lambda_h]; \\
P(\hat{r}^b_{\lambda} \ne r^{\ast}) \le 2\exp\{-\theta^2 (\sqrt{r_{\mathrm{x}^b}} + \sqrt{q})^2 /8\}, \quad  &\text{when } \lambda \in  [\lambda_h, \lambda_u].
\end{align*}
Let $K_1 = |\Lambda \cap [\lambda_m, \lambda_h]|$ and $K_2 = |\Lambda \cap [\lambda_h, \lambda_u]|$. Then,
\begin{align}
& P(\hat{r}^b_{\hat\lambda} \ne r^{\ast}) \\
\le {} & K_1 (\delta + 2 \exp\{-\theta^2 (\sqrt{r_{\mathrm{x}^b}} + \sqrt{q})^2 /8\}) + 2 K_2 \exp\{-\theta^2 (\sqrt{r_{\mathrm{x}^b}} + \sqrt{q})^2 /8\} + P(\hat\lambda \notin [\lambda_m, \lambda_u]) \notag \\
\le {} & K \delta + 2 K \exp\{-\theta^2 (\sqrt{r_{\mathrm{x}^b}} + \sqrt{q})^2 /8\} +  6 K \exp[ -n \delta^2 / \{10368 (r_{\mathrm{x}} \wedge q)^4 b \}],\label{inconsistency.probability}
\end{align}
where $\max[32 (r_{\mathrm{x}^b} \wedge q)^2 \exp\{-\theta^2 (\sqrt{r_{\mathrm{x}^b}} + \sqrt{q})^2 /8\}, 48 (r_{\mathrm{x}^b} \wedge q)^2/(N-1)] \le \delta < 1/2$.

Under the conditions that for a given $\alpha > 0$ and a given $\beta > 0$,
\begin{align*}
(r_{\mathrm{x}^b} \wedge q)^2 \exp\{-\theta^2 (\sqrt{r_{\mathrm{x}^b}} + \sqrt{q})^2/8\} &= o(n^{-\alpha}),\\
(r_{\mathrm{x}} \wedge q)^4 b/n &= o(n^{-\beta}).
\end{align*}
Then, we can choose $\delta = n^{-(\alpha \wedge \beta)/2}$ such that 
$\max[32 (r_{\mathrm{x}^b} \wedge q)^2 \exp\{-\theta^2 (\sqrt{r_{\mathrm{x}^b}} + \sqrt{q})^2 /8\}, 48 (r_{\mathrm{x}^b} \wedge q)^2/(N-1)] \le \delta < 1/2$ when $n$ is large enough, and 
\[ 6 K \exp[ -n \delta^2 / \{10368 (r_{\mathrm{x}} \wedge q)^4 b \}] \le 6 K \exp[ -n^{1-\beta} /\{10368 (r_{\mathrm{x}} \wedge q)^4 b \} ] \to 0.\]
Therefore, the probability bound in (\ref{inconsistency.probability}) tends to zero when $r_{\mathrm{x}^b} + q \to \infty$ because $K$ is assumed to be fixed.

\section{Additional Simulation Results}
\renewcommand{\thesubsection}{\Alph{section}.\arabic{subsection}}
\subsection{Prediction Performance}

\subsubsection{Impact of rank estimation on prediction}

Rank estimation is the focus of this paper as better rank estimation can indeed lead to better prediction. Hereby, we perform a simulation study to illustrate the impact of rank estimation on prediction. We adopt the same simulation settings in the main text to generate the training and test data sets, where the sample size of the test data is set to be the same as that of the training data. For each setting, we replicate 100 times. Given that the simulation results are similar, we only present the results of three settings for the purpose of illustration, which are: 
\begin{itemize}
    \item [(1)] $\rho=0.9$ and SNR=1.57 under Model \uppercase\expandafter{\romannumeral1};
    \item [(2)] $\rho=0.1$ and SNR=1.15 under Model \uppercase\expandafter{\romannumeral2};
    \item [(3)] $\rho=0.5$ and SNR=3.09 under Model \uppercase\expandafter{\romannumeral2}.
\end{itemize}

For each given rank $\hat r$ in $\{1,2,\dots,(p\wedge q)-1\}$, we derive the estimated coefficient matrix $\hat{C}$ using the ridge estimation method \citep{izenman2008modern}. Denoting the test data as $\{\widetilde{X},\widetilde{Y}\}$, we calculate the mean squared prediction error (MSPE) as
$$
\mathrm{MSPE}=100\times\frac{\lVert\widetilde{  Y}-\widetilde{ X}\hat{ C}\rVert_{\mathrm{F}}^{2}}{nq}.
$$

The simulation results are summarized in Figure~\ref{fig:rd_pe}, in which the top, middle, and bottom panels correspond to settings (1)--(3). From Figure~\ref{fig:rd_pe}, it is obvious that a better rank estimate leads to a better prediction. Overall, the true rank estimate always corresponds to the smallest mean MSPE, and the MSPE is larger than the one with the true rank when the rank is under-estimated or over-estimated. In particular, when the estimated rank is smaller than the true rank, the information extracted by the model is insufficient, resulting in a poor prediction accuracy. When the estimated rank is larger than the true rank, the extracted signal is mixed with some noises, resulting in a slightly lower prediction accuracy. Furthermore, as the discrepancy between the estimated rank and the true rank becomes larger, the MSPE increases quickly especially in the underestimation scenario.

\begin{figure}[htbp]
\centering 
\includegraphics[width=0.75\textwidth]{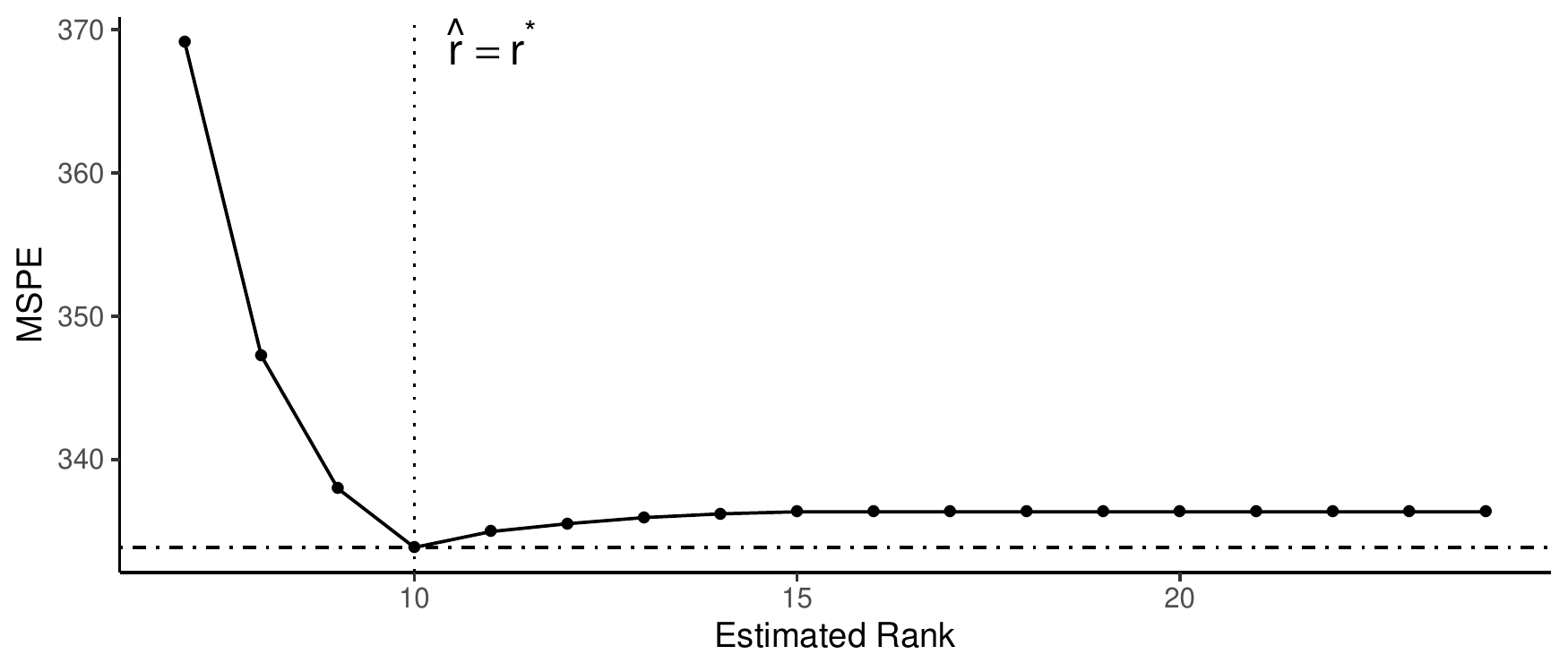}		        \includegraphics[width=0.75\textwidth]{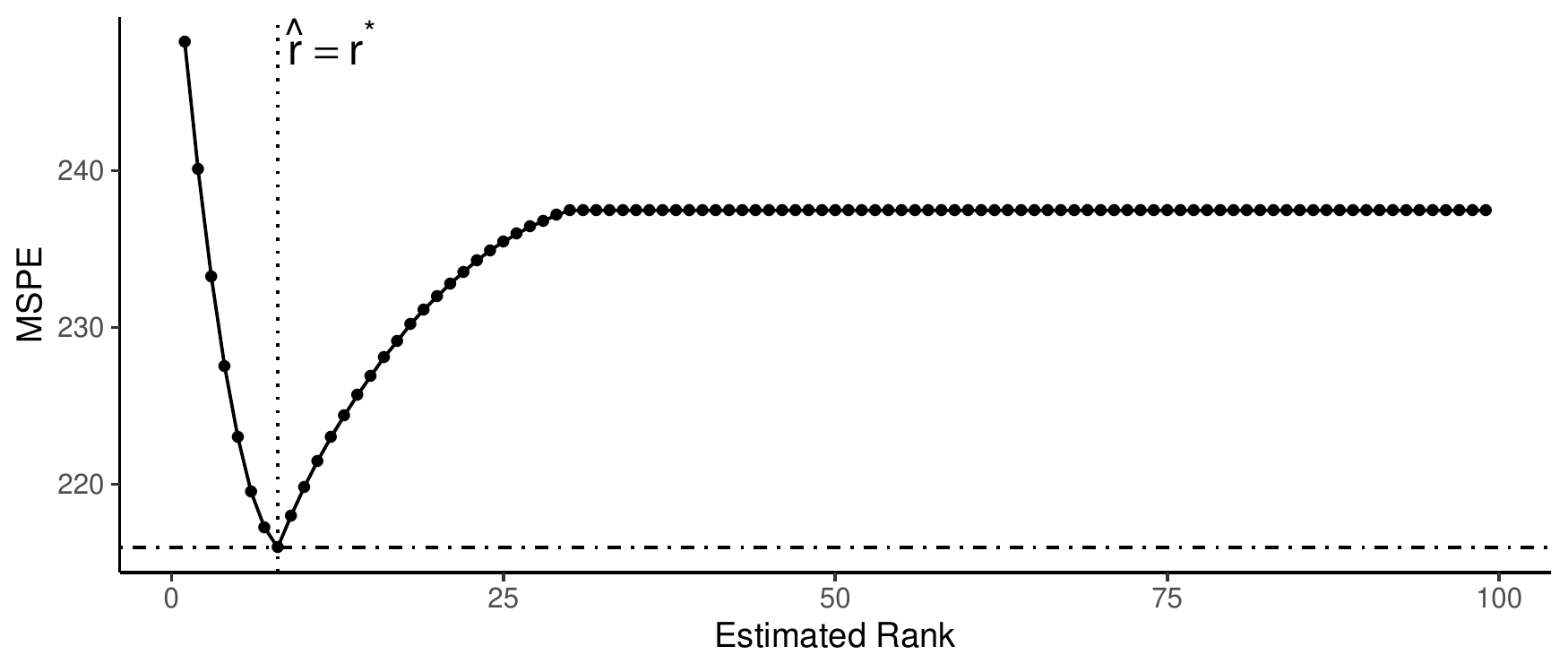}
\includegraphics[width=0.75\textwidth]{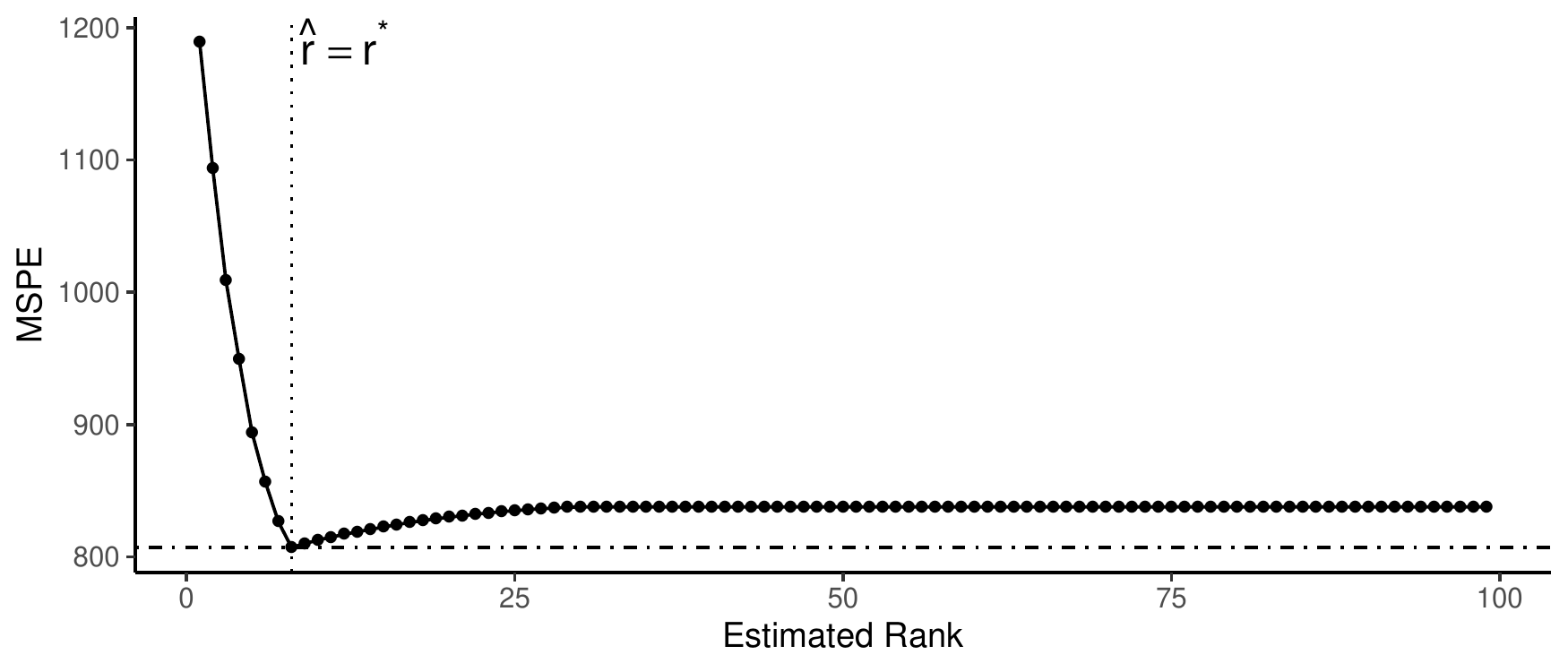}
\caption{Line graph of prediction error versus estimated rank. The dots represent the mean value of MSPE corresponding to each estimated rank. The ordinate of the dotdash line represents the minimum value of the mean MSPE, and the dotted line represents the true rank. 
}
\label{fig:rd_pe}
\end{figure}


\subsubsection{Prediction performance}

In this subsection, we evaluate the prediction performance of different rank determination methods. We adopt the same simulation settings in the main text to generate the training and test data sets. For each estimated coefficient matrix $\hat{C}$, we calculate the mean squared prediction error (MSPE) on the test data $\{\widetilde{X},\widetilde{Y}\}$ with a refitted estimator \citep{izenman2008modern}. The simulation results are summarized in Table~\ref{tab:pe}. Our method still performs well with the smallest MSPE in most cases although the difference of the MSPE's is small. This is partially because that a slightly overestimated rank may not cause a significant increase of MSPE, while an underestimated rank often leads to a substantial increase of MSPE (see Section B.1.1). 
\iftrue
\begin{table}[htbp]
	\caption{The prediction performance in the simulation study.}
	\label{tab:pe}
	\centering
		\begin{tabular}{cccccccc}
			\hline
			\hline
			 SNR & AIC & BIC & GIC & BICP & GCV & CV  & StARS-RRR\\
			\hline
			\multicolumn{8}{c}{Model \uppercase\expandafter{\romannumeral1}, $\rho = 0.1$} \\
			1.07 & 237.49 & 238.42 & 240.13 & 239.57 & 237.46 & 237.48 & 237.45 \\ 
              1.60 & 404.68 & 405.12 & 406.10 & 405.67 & 404.67 & 404.68 & 404.61 \\ 
              1.85 & 505.66 & 505.88 & 506.75 & 506.48 & 505.65 & 505.66 & 505.56 \\ 
              2.14 & 638.93 & 639.01 & 639.56 & 639.34 & 638.92 & 638.91 & 638.79 \\ 
              2.49 & 832.29 & 832.25 & 832.56 & 832.50 & 832.27 & 832.27 & 832.14 \\ 
              3.03 & 1178.12 & 1177.96 & 1178.14 & 1178.10 & 1178.10 & 1178.10 & 1177.99 \\
              \hline
              \multicolumn{8}{c}{Model \uppercase\expandafter{\romannumeral1}, $\rho = 0.5$} \\
              
              1.10 & 236.38 & 237.28 & 239.09 & 238.47 & 236.35 & 236.39 & 236.33 \\ 
              1.26 & 276.80 & 277.57 & 279.08 & 278.54 & 276.78 & 276.81 & 276.65 \\ 
              1.57 & 373.82 & 373.99 & 375.30 & 374.68 & 373.79 & 373.81 & 373.72 \\ 
              2.20 & 632.61 & 632.60 & 633.06 & 632.91 & 632.57 & 632.61 & 632.44 \\ 
              2.52 & 794.34 & 794.27 & 794.64 & 794.53 & 794.32 & 794.34 & 794.19 \\ 
              2.99 & 1077.40 & 1077.28 & 1077.43 & 1077.38 & 1077.38 & 1077.40 & 1077.26 \\ 
              \hline
              \multicolumn{8}{c}{Model \uppercase\expandafter{\romannumeral1}, $\rho = 0.9$} \\
              1.08 & 217.00 & 218.01 & 220.14 & 219.43 & 216.96 & 217.01 & 216.97 \\ 
              1.24 & 251.14 & 252.00 & 253.78 & 253.24 & 251.11 & 251.13 & 251.04 \\ 
              1.55 & 333.13 & 333.74 & 334.89 & 334.53 & 333.11 & 333.12 & 333.01 \\ 
              2.09 & 520.52 & 520.53 & 521.52 & 521.08 & 520.47 & 520.51 & 520.26 \\ 
              2.71 & 803.05 & 802.84 & 803.26 & 803.21 & 803.02 & 803.01 & 802.79 \\ 
              3.10 & 1016.65 & 1016.43 & 1016.89 & 1016.60 & 1016.62 & 1016.64 & 1016.41 \\ 
              \hline\hline
              \multicolumn{8}{c}{Model \uppercase\expandafter{\romannumeral2}, $\rho = 0.1$} \\
              1.16 & 216.13 & 220.40 & 244.57 & 244.08 & 214.26 & 214.13 & 216.18 \\ 
              1.45 & 276.26 & 277.72 & 322.79 & 319.01 & 274.37 & 274.20 & 274.12 \\ 
              1.73 & 350.09 & 349.45 & 412.87 & 401.34 & 348.09 & 347.90 & 347.87 \\ 
              2.02 & 437.40 & 435.71 & 509.47 & 486.28 & 435.30 & 435.11 & 435.09 \\ 
              2.60 & 652.50 & 649.91 & 705.01 & 666.10 & 650.11 & 649.90 & 649.89 \\ 
              3.18 & 921.15 & 918.44 & 924.64 & 920.76 & 918.70 & 918.46 & 918.44 \\
              \hline
              \multicolumn{8}{c}{Model \uppercase\expandafter{\romannumeral2}, $\rho = 0.5$} \\
              1.12 & 201.41 & 206.85 & 239.20 & 238.29 & 198.72 & 198.58 & 202.40 \\ 
              1.40 & 251.10 & 253.57 & 311.83 & 305.16 & 248.50 & 248.25 & 248.40 \\ 
              1.68 & 312.25 & 312.04 & 392.34 & 377.91 & 309.57 & 309.31 & 309.26 \\ 
              2.09 & 425.08 & 422.60 & 518.26 & 481.77 & 422.25 & 422.00 & 422.03 \\ 
              2.51 & 563.07 & 559.91 & 631.77 & 587.19 & 560.18 & 559.90 & 559.92 \\ 
              3.07 & 786.30 & 782.79 & 793.90 & 786.27 & 783.17 & 782.80 & 782.84 \\ 
              \hline
              \multicolumn{8}{c}{Model \uppercase\expandafter{\romannumeral2}, $\rho = 0.9$} \\
              1.05 & 170.88 & 175.56 & 247.20 & 235.18 & 165.69 & 165.46 & 167.53 \\ 
              1.24 & 188.74 & 191.17 & 277.27 & 257.08 & 183.59 & 183.19 & 183.14 \\ 
              1.52 & 222.05 & 221.85 & 313.44 & 283.26 & 216.62 & 216.08 & 215.89 \\ 
              2.00 & 293.59 & 288.82 & 360.26 & 332.62 & 287.77 & 287.21 & 287.19 \\ 
              2.48 & 384.67 & 378.16 & 415.47 & 400.91 & 378.65 & 378.01 & 378.02 \\ 
              3.05 & 519.68 & 512.57 & 525.72 & 519.54 & 513.33 & 512.62 & 512.66 \\ 
			\hline
			\hline
	\end{tabular}
\end{table}
\fi

\subsection{Effective Rank}

On the one hand, the true rank of a $p\times q$ matrix $C$ is defined as
$$
r^*=\text{rank}(C)=\sum_{i=1}^{p\wedge q}I(d_i(C) > 0),
$$
where $d_i(C)$ is the $i$th singular value of $C$, $i=1,\ldots,p\wedge q$. On the other hand, according to \cite{bunea2011optimal}, the effective rank of the coefficient matrix $C$ in reduced-rank regression is defined as 
$$ 
s=\sum_{i=1}^{p \wedge q}I(d_i(XC)>d_1(PE)),
$$
where $X$, $P$, and $E$ are the design matrix, projection matrix, and error matrix in reduced-rank regression, respectively.

For any two matrices $A, B$ with suitable dimensions, we have that
$$
\text{rank}(AB)\leq \text{rank}(A)\wedge\text{rank}(B).
$$
Therefore,
$$
s=\sum_{i=1}^{p \wedge q}I(d_i(XC)>d_1(PE))\leq\text{rank}(XC)\leq\text{rank}(C)=r^*.
$$ 
Furthermore, by the definition of the signal-to-noise ratio: $\text{SNR}=d_{r^*}(XC)/d_1(PE)$, we have $d_{r^*}(XC)>d_1(PE)$ when $\text{SNR}>1 $. Since $d_{i}(XC)\geq d_{r^*}(XC)$ holds for any $i=1,\ldots ,r^*$, we know that $s \ge r^*$ when $\text{SNR}>1 $. In conclusion, the effective rank is equal to the true rank when $\text{SNR}>1$.

However, in the simulation studies of the main text, the effective rank \citep{bunea2011optimal} and the true rank of the coefficient matrix are not necessarily identical. This is because not all simulated data have a $\text{SNR}>1$, although the average SNR across the replicates is reported to be greater than 1 in all simulation settings (see the SNR column in Table 2).

Therefore, we additionally report the effective-rank recovery ratio. Table~\ref{tab:effr}~reports the effective-rank recovery ratio under the same simulations settings as in the main text. We can see that StARS-RRR has the highest effective-rank recovery ratio in all cases.

\iftrue
\begin{table}[htbp]
	\caption{The effective-rank recovery ratio in the simulation study.}
	\label{tab:effr}
	\centering
		\begin{tabular}{cccccccc}
			\hline
			\hline
			 SNR & AIC & BIC & GIC & BICP & GCV & CV  & StARS-RRR\\
			\hline
			\multicolumn{8}{c}{Model \uppercase\expandafter{\romannumeral1}, $\rho = 0.1$} \\
            		
              1.07 & 0.51 & 0.45 & 0.15 & 0.20 & 0.52 & 0.52 & 0.80 \\ 
              1.60 & 0.78 & 0.72 & 0.45 & 0.51 & 0.80 & 0.79 & 0.93 \\ 
              1.85 & 0.83 & 0.80 & 0.57 & 0.64 & 0.84 & 0.84 & 0.96 \\ 
              2.14 & 0.85 & 0.88 & 0.70 & 0.76 & 0.86 & 0.86 & 0.99 \\ 
              2.49 & 0.85 & 0.95 & 0.84 & 0.87 & 0.87 & 0.87 & 0.99 \\ 
              3.03 & 0.85 & 0.99 & 0.93 & 0.96 & 0.87 & 0.87 & 0.99 \\ 
              
            \hline
			\multicolumn{8}{c}{Model \uppercase\expandafter{\romannumeral1}, $\rho = 0.5$} \\
              1.10 & 0.57 & 0.53 & 0.19 & 0.25 & 0.58 & 0.55 & 0.83 \\ 
              1.26 & 0.64 & 0.60 & 0.29 & 0.37 & 0.65 & 0.64 & 0.88 \\ 
              1.57 & 0.77 & 0.71 & 0.44 & 0.51 & 0.78 & 0.78 & 0.91 \\ 
              2.20 & 0.84 & 0.89 & 0.71 & 0.76 & 0.86 & 0.85 & 0.98 \\ 
              2.52 & 0.86 & 0.94 & 0.81 & 0.85 & 0.87 & 0.87 & 0.99 \\ 
              2.99 & 0.86 & 0.97 & 0.91 & 0.93 & 0.88 & 0.86 & 0.99 \\ 
			\hline
			\multicolumn{8}{c}{Model \uppercase\expandafter{\romannumeral1}, $\rho = 0.9$} \\
              1.08 & 0.52 & 0.50 & 0.20 & 0.26 & 0.53 & 0.55 & 0.81 \\ 
              1.24 & 0.62 & 0.58 & 0.27 & 0.34 & 0.64 & 0.64 & 0.86 \\ 
              1.55 & 0.73 & 0.71 & 0.44 & 0.52 & 0.75 & 0.75 & 0.94 \\ 
              2.09 & 0.82 & 0.87 & 0.67 & 0.74 & 0.85 & 0.84 & 0.98 \\ 
              2.71 & 0.85 & 0.95 & 0.85 & 0.89 & 0.87 & 0.87 & 0.99 \\ 
              3.10 & 0.85 & 0.98 & 0.91 & 0.94 & 0.87 & 0.86 & 1.00 \\
            \hline
			\hline
			\multicolumn{8}{c}{Model \uppercase\expandafter{\romannumeral2}, $\rho = 0.1$} \\              
              1.16 & 0.37 & 0.00 & 0.00 & 0.00 & 0.71 & 0.75 & 0.75 \\ 
              1.45 & 0.48 & 0.13 & 0.00 & 0.00 & 0.89 & 0.95 & 0.99 \\ 
              1.73 & 0.50 & 0.54 & 0.00 & 0.00 & 0.93 & 0.99 & 1.00 \\ 
              2.02 & 0.50 & 0.83 & 0.00 & 0.00 & 0.93 & 0.99 & 1.00 \\ 
              2.60 & 0.47 & 0.99 & 0.05 & 0.24 & 0.93 & 0.99 & 1.00 \\ 
              3.18 & 0.45 & 1.00 & 0.65 & 0.80 & 0.92 & 0.99 & 1.00 \\ 
			\hline
			\multicolumn{8}{c}{Model \uppercase\expandafter{\romannumeral2}, $\rho = 0.5$} \\              
              1.12 & 0.35 & 0.01 & 0.00 & 0.00 & 0.69 & 0.71 & 0.71 \\ 
              1.40 & 0.47 & 0.07 & 0.00 & 0.00 & 0.87 & 0.92 & 0.94 \\ 
              1.68 & 0.51 & 0.45 & 0.00 & 0.00 & 0.92 & 0.98 & 1.00 \\ 
              2.09 & 0.52 & 0.88 & 0.00 & 0.00 & 0.94 & 0.99 & 1.00 \\ 
              2.51 & 0.52 & 0.99 & 0.02 & 0.17 & 0.93 & 0.99 & 1.00 \\ 
              3.07 & 0.48 & 1.00 & 0.55 & 0.75 & 0.91 & 1.00 & 1.00 \\ 
			\hline
			\multicolumn{8}{c}{Model \uppercase\expandafter{\romannumeral2}, $\rho = 0.9$} \\              
              1.05 & 0.26 & 0.02 & 0.00 & 0.00 & 0.51 & 0.53 & 0.63 \\ 
              1.24 & 0.41 & 0.05 & 0.00 & 0.00 & 0.77 & 0.82 & 0.87 \\ 
              1.52 & 0.47 & 0.27 & 0.00 & 0.00 & 0.88 & 0.94 & 0.98 \\ 
              2.00 & 0.49 & 0.82 & 0.00 & 0.01 & 0.91 & 0.99 & 1.00 \\ 
              2.48 & 0.49 & 0.98 & 0.15 & 0.25 & 0.90 & 0.99 & 1.00 \\ 
              3.05 & 0.47 & 1.00 & 0.61 & 0.74 & 0.90 & 0.99 & 1.00 \\ 
			\hline
			\hline
	\end{tabular}
\end{table}
\fi

\subsection{Sensitivity Analyses}

The choice of hyperparameters is important for StARS-RRR. Suggested by the referees, we conduct sensitivity analyses for the hyperparameters in StARS-RRR, including the threshold $\eta$, the number of subsamples $N$, and the size of subsamples $b$. Below are the values of $\eta$, $N$, and $b$ that are examined in each sensitivity analysis.
\begin{itemize}
    \item Sensitivity analysis of $\eta$: $\eta_1=0.001$, $\eta_2=0.01$, $\eta_3=0.02$, $\eta_4=0.05$
    \item Sensitivity analysis of $N$: $N_1=25$, $N_2=50$, $N_3=100$, $N_4=200$
    \item Sensitivity analysis of $b$:
    \begin{itemize}
    \item Model I: $b_1=\lfloor 10\sqrt{n}\rfloor$, $b_2=0.5n$, $b_3=0.7n$, $b_4=0.9n$
    \item Model II: $b_1=0.4n$, $b_2=0.5n$, $b_3=0.7n$, $b_4=0.9n$
    \end{itemize}
\end{itemize}

When one hyperparameter is examined in each sensitivity analysis, the other two hyperparameters are fixed at their values that are used in the main text: $\eta = 0.001$, $b = 0.7n$, and $N = 100$.

Except the choices of $\eta$, $N$, and $b$, the other simulation settings are exactly the same as the simulation studies in the main text. 

\subsubsection{Sensitivity analysis for $\eta$}

The results of the sensitivity analysis of $\eta$ are summarzied in Table~\ref{tab:eta}. We find that when SNR is very small, a slightly larger value of $\eta$ could lead to a better rank estimation when the signal-to-noise ratio is very small. A larger value of $\eta$ leads to a smaller tuning parameter, increasing the the possibility of an overestimated rank. The simulation studies in the main text show that StARS-RRR tends to underestimate the rank when the SNR is very small. Therefore, a slightly larger value of $\eta$ could alleviate this tendency, resulting in a more accurately estimated rank. By contrast, when SNR is sufficiently large, the smallest value of $\eta$ seems to outperform the others by a small margin.

In practice when we have no prior knowledge of the signal-to-noise ratio, we recommend using a small value of $\eta$ to be more conservative on the instability measure.


	\begin{table}[htbp]
	\centering
	\caption{Rank recovery (left), underestimate (middle), and overestimate (right) ratios (in percentage) in the sensitive analysis of $ \eta $.}\label{tab:eta}
	\begin{tabular}{cccccc}
		\hline
		\hline
		s & SNR & $ \eta_1=0.001 $ & $ \eta_2=0.01 $ & $ \eta_3=0.02 $ & $ \eta_4=0.05 $ \\ 
		\hline
		\multicolumn{6}{c}{Model \uppercase\expandafter{\romannumeral1}, $\rho = 0.1$} \\
		40 & 1.42 & (86,14,0) & (90,9,0) & (91,8,1) & (93,5,2) \\ 
		45 & 1.6 & (92,8,0) & (95,5,0) & (96,4,1) & (95,3,2) \\ 
		52 & 1.85 & (95,4,0) & (97,2,0) & (97,2,1) & (96,1,2) \\ 
		60 & 2.14 & (98,2,0) & (98,1,0) & (98,1,1) & (97,1,2) \\ 
		70 & 2.49 & (98,1,0) & (98,1,0) & (98,1,1) & (98,0,2) \\ 
		85 & 3.03 & (99,0,0) & (99,0,1) & (99,0,1) & (98,0,2) \\ 
		\hline
		\multicolumn{6}{c}{Model \uppercase\expandafter{\romannumeral1}, $\rho = 0.5$} \\
		35 & 1.1 & (63,37,0) & (70,29,1) & (72,26,1) & (79,18,3) \\ 
		40 & 1.26 & (73,26,1) & (79,20,1) & (81,17,1) & (86,11,3) \\ 
		50 & 1.57 & (89,10,1) & (93,6,1) & (93,6,1) & (93,4,3) \\ 
		70 & 2.2 & (97,3,1) & (97,2,1) & (97,1,1) & (96,1,3) \\ 
		80 & 2.52 & (99,1,1) & (98,1,1) & (98,1,1) & (97,0,3) \\ 
		95 & 2.99 & (99,1,1) & (98,1,1) & (98,0,1) & (97,0,3) \\ 
		\hline
		\multicolumn{6}{c}{Model \uppercase\expandafter{\romannumeral1}, $\rho = 0.9$} \\
		70 & 1.08 & (63,37,0) & (67,32,0) & (72,27,1) & (76,19,5) \\ 
		80 & 1.24 & (73,27,0) & (79,21,0) & (82,17,1) & (82,13,5) \\ 
		100 & 1.55 & (86,14,0) & (90,10,0) & (92,7,1) & (91,4,5) \\ 
		135 & 2.09 & (97,3,0) & (97,2,1) & (97,2,1) & (95,1,5) \\ 
		175 & 2.71 & (99,1,0) & (99,0,1) & (98,0,1) & (95,0,5) \\ 
		200 & 3.1 & (99,0,0) & (99,0,1) & (98,0,1) & (95,0,5) \\
		\hline
		\hline
		\multicolumn{6}{c}{Model \uppercase\expandafter{\romannumeral2}, $\rho = 0.1$} \\
		8 & 1.16 & (77,23,0) & (88,12,0) & (91,8,0) & (93,4,3) \\ 
		10 & 1.45 & (98,2,0) & (99,1,0) & (99,1,0) & (97,0,3) \\ 
		12 & 1.73 & (100,0,0) & (100,0,0) & (100,0,0) & (97,0,3) \\ 
		14 & 2.02 & (100,0,0) & (100,0,0) & (100,0,0) & (97,0,3) \\ 
		18 & 2.6 & (100,0,0) & (100,0,0) & (100,0,0) & (97,0,3) \\ 
		22 & 3.18 & (100,0,0) & (100,0,0) & (100,0,0) & (97,0,3) \\ 
		\hline
		\multicolumn{6}{c}{Model \uppercase\expandafter{\romannumeral2}, $\rho = 0.5$} \\
		8 & 1.12 & (69,30,0) & (82,17,0) & (89,11,1) & (93,5,2) \\ 
		10 & 1.4 & (95,4,0) & (98,2,0) & (98,1,1) & (98,0,2) \\ 
		12 & 1.68 & (100,0,0) & (100,0,0) & (100,0,0) & (97,0,3) \\ 
		15 & 2.09 & (100,0,0) & (100,0,0) & (100,0,0) & (97,0,3) \\ 
		18 & 2.51 & (100,0,0) & (100,0,0) & (100,0,0) & (97,0,3) \\ 
		22 & 3.07 & (100,0,0) & (100,0,0) & (100,0,0) & (97,0,3) \\
		\hline
		\multicolumn{6}{c}{Model \uppercase\expandafter{\romannumeral2}, $\rho = 0.9$} \\
		11 & 1.05 & (58,42,0) & (73,27,0) & (77,21,1) & (89,8,3) \\ 
		13 & 1.24 & (86,14,0) & (94,5,1) & (96,2,1) & (96,1,3) \\ 
		16 & 1.52 & (99,1,0) & (99,0,1) & (99,0,1) & (97,0,3) \\ 
		21 & 2 & (100,0,0) & (99,0,1) & (99,0,1) & (96,0,4) \\ 
		26 & 2.48 & (100,0,0) & (99,0,1) & (99,0,1) & (97,0,3) \\ 
		32 & 3.05 & (100,0,0) & (99,0,1) & (99,0,1) & (97,0,3) \\ 
		\hline
		\hline
	\end{tabular}
\end{table}

\subsubsection{Sensitivity analysis for $N$}

The results of the sensitivity analysis for $N$ are summarized in Table ~\ref{tab:bags}. We find that the rank estimation is not very sensitive to the choice of $N$ as long as it is not too small, say, 25. We recommend using 50 or 100, which balances the stability of estimating the sample variance of the estimated ranks and the computational cost.

\begin{table}[htbp]
	\centering
	\caption{Rank recovery (left), underestimate (middle), and overestimate (right) ratios (in percentage) in the sensitive analysis of $ N $.}
	\label{tab:bags}
	\begin{tabular}{cccccc}
		\hline
		\hline
		s & SNR & $ N_1=25 $ & $ N_2=50 $ & $ N_3=100 $ & $ N_4=200 $ \\ 		   
		\hline
		\multicolumn{6}{c}{Model \uppercase\expandafter{\romannumeral1}, $\rho = 0.1$} \\
		40 & 1.48 & (93,3,4) & (93,7,0) & (89,11,0) & (86,14,0) \\ 
		45 & 1.67 & (92,2,6) & (95,5,0) & (93,7,0) & (93,7,0) \\ 
		52 & 1.93 & (94,1,5) & (99,1,0) & (96,4,0) & (95,5,0) \\ 
		60 & 2.22 & (95,0,5) & (100,0,0) & (100,0,0) & (100,0,0) \\ 
		70 & 2.6 & (93,0,7) & (100,0,0) & (100,0,0) & (100,0,0) \\ 
		85 & 3.15 & (96,0,4) & (100,0,0) & (100,0,0) & (100,0,0) \\
		\hline
		\multicolumn{6}{c}{Model \uppercase\expandafter{\romannumeral1}, $\rho = 0.5$} \\
		35 & 1.15 & (72,16,12) & (71,29,0) & (64,36,0) & (61,39,0) \\ 
		40 & 1.32 & (82,7,11) & (84,16,0) & (77,23,0) & (74,26,0) \\ 
		50 & 1.65 & (86,4,10) & (93,6,1) & (92,8,0) & (90,10,0) \\ 
		70 & 2.31 & (91,0,9) & (99,0,1) & (98,2,0) & (97,3,0) \\ 
		80 & 2.64 & (90,0,10) & (99,0,1) & (100,0,0) & (100,0,0) \\ 
		95 & 3.13 & (91,0,9) & (100,0,0) & (100,0,0) & (100,0,0) \\
		\hline
		\multicolumn{6}{c}{Model \uppercase\expandafter{\romannumeral1}, $\rho = 0.9$} \\
		70 & 1.13 & (72,18,10) & (69,31,0) & (66,34,0) & (61,39,0) \\ 
		80 & 1.29 & (76,12,12) & (79,20,1) & (76,24,0) & (70,30,0) \\ 
		100 & 1.62 & (86,5,9) & (89,10,1) & (85,15,0) & (84,16,0) \\ 
		135 & 2.18 & (88,1,11) & (98,1,1) & (98,2,0) & (96,4,0) \\ 
		175 & 2.83 & (89,1,10) & (97,1,2) & (99,1,0) & (99,1,0) \\ 
		200 & 3.23 & (89,0,11) & (98,0,2) & (99,1,0) & (99,1,0) \\  
		\hline
		\hline
		\multicolumn{6}{c}{Model \uppercase\expandafter{\romannumeral2}, $\rho = 0.1$} \\
		8 & 1.15 & (91,6,3) & (86,14,0) & (78,22,0) & (67,33,0) \\ 
		10 & 1.44 & (96,0,4) & (97,3,0) & (96,4,0) & (96,4,0) \\ 
		12 & 1.73 & (97,0,3) & (100,0,0) & (100,0,0) & (100,0,0) \\ 
		14 & 2.02 & (97,0,3) & (100,0,0) & (100,0,0) & (100,0,0) \\ 
		18 & 2.59 & (98,0,2) & (100,0,0) & (100,0,0) & (100,0,0) \\ 
		22 & 3.17 & (98,0,2) & (100,0,0) & (100,0,0) & (100,0,0) \\ 
		\hline
		\multicolumn{6}{c}{Model \uppercase\expandafter{\romannumeral2}, $\rho = 0.5$} \\
		8 & 1.1 & (90,9,1) & (78,22,0) & (64,36,0) & (51,49,0) \\ 
		10 & 1.38 & (99,0,1) & (98,2,0) & (96,4,0) & (95,5,0) \\ 
		12 & 1.65 & (99,0,1) & (100,0,0) & (100,0,0) & (100,0,0) \\ 
		15 & 2.07 & (99,0,1) & (100,0,0) & (100,0,0) & (100,0,0) \\ 
		18 & 2.48 & (99,0,1) & (100,0,0) & (100,0,0) & (100,0,0) \\ 
		22 & 3.03 & (100,0,0) & (100,0,0) & (100,0,0) & (100,0,0) \\  
		\hline
		\multicolumn{6}{c}{Model \uppercase\expandafter{\romannumeral2}, $\rho = 0.9$} \\
		11 & 1.04 & (81,17,2) & (72,28,0) & (63,37,0) & (49,51,0) \\ 
		13 & 1.23 & (93,4,3) & (91,9,0) & (83,17,0) & (79,21,0) \\ 
		16 & 1.51 & (97,0,3) & (100,0,0) & (98,2,0) & (97,3,0) \\ 
		21 & 1.98 & (97,0,3) & (100,0,0) & (100,0,0) & (100,0,0) \\ 
		26 & 2.45 & (97,0,3) & (100,0,0) & (100,0,0) & (100,0,0) \\ 
		32 & 3.02 & (97,0,3) & (100,0,0) & (100,0,0) & (100,0,0) \\ 
		\hline
		\hline
	\end{tabular}
\end{table}

\subsubsection{Sensitivity analysis for $b$}

The results of the sensitivity analysis of $b$ are summarized in Table ~\ref{tab:sub}. The accuracy of the rank estimation is closely related to the size of the subsample. From the results, the overall performance of the rank estimation is the best when $b = 0.7n$ considering both Models I and II. In fact, the performance of StARS-RRR is satisfactory as long as $b$ is not too small, such as $b=\lfloor 10\sqrt{n}\rfloor $ or $b = 0.4n$.

\begin{table}[htbp]
	\centering
	\caption{Rank recovery (left), underestimate (middle), and overestimate (right) ratios (in percentage) in the sensitive analysis of $ b $.}
	\label{tab:sub}
	\begin{tabular}{cccccc}
		\hline
		\hline
		s & SNR & $ b_1=\lfloor 10\sqrt{n}\rfloor $ & $ b_2=0.5n $ & $ b_3=0.7n $ & $ b_4=0.9n$ \\ 		   
		\hline
		\multicolumn{6}{c}{Model \uppercase\expandafter{\romannumeral1}, $\rho = 0.1$} \\
		40 & 1.48 & (80,19,1) & (78,21,1) & (89,11,0) & (84,1,15) \\ 
		45 & 1.67 & (88,11,1) & (90,10,0) & (93,7,0) & (84,0,16) \\ 
		52 & 1.93 & (95,5,0) & (93,7,0) & (96,4,0) & (87,0,13) \\ 
		60 & 2.22 & (99,1,0) & (99,1,0) & (100,0,0) & (84,0,16) \\ 
		70 & 2.6 & (100,0,0) & (100,0,0) & (100,0,0) & (84,0,16) \\ 
		85 & 3.15 & (100,0,0) & (100,0,0) & (100,0,0) & (86,0,14) \\ 
		\hline
		\multicolumn{6}{c}{Model \uppercase\expandafter{\romannumeral1}, $\rho = 0.5$} \\
		35 & 1.15 & (46,53,1) & (50,50,0) & (64,36,0) & (80,10,10) \\ 
		40 & 1.32 & (70,29,1) & (69,31,0) & (77,23,0) & (86,4,10) \\ 
		50 & 1.65 & (87,13,0) & (86,14,0) & (92,8,0) & (88,0,12) \\ 
		70 & 2.31 & (99,1,0) & (99,1,0) & (98,2,0) & (86,0,14) \\ 
		80 & 2.64 & (100,0,0) & (100,0,0) & (100,0,0) & (85,0,15) \\ 
		95 & 3.13 & (100,0,0) & (100,0,0) & (100,0,0) & (87,0,13) \\ 
		\hline
		\multicolumn{6}{c}{Model \uppercase\expandafter{\romannumeral1}, $\rho = 0.9$} \\
		70 & 1.13 & (44,54,2) & (53,46,1) & (66,34,0) & (78,10,12) \\ 
		80 & 1.29 & (63,35,2) & (63,36,1) & (76,24,0) & (80,7,13) \\ 
		100 & 1.62 & (79,19,2) & (83,16,1) & (85,15,0) & (87,2,11) \\ 
		135 & 2.18 & (94,4,2) & (95,4,1) & (98,2,0) & (87,1,12) \\ 
		175 & 2.83 & (97,1,2) & (97,2,1) & (99,1,0) & (89,1,10) \\ 
		200 & 3.23 & (97,1,2) & (98,1,1) & (99,1,0) & (90,0,10) \\ 
		\hline
		\hline
		s & SNR & $ b_1=0.4n $ & $ b_2=0.5n $ & $ b_3=0.7n $ & $ b_4=0.9n$ \\
		\hline
		\multicolumn{6}{c}{Model \uppercase\expandafter{\romannumeral2}, $\rho = 0.1$} \\
		8 & 1.15 & (0,100,0) & (8,90,2) & (78,22,0) & (97,3,0) \\ 
		10 & 1.44 & (5,94,1) & (47,51,2) & (96,4,0) & (100,0,0) \\ 
		12 & 1.73 & (42,57,1) & (87,11,2) & (100,0,0) & (100,0,0) \\ 
		14 & 2.02 & (74,26,0) & (96,2,2) & (100,0,0) & (100,0,0) \\ 
		18 & 2.59 & (98,2,0) & (98,0,2) & (100,0,0) & (100,0,0) \\ 
		22 & 3.17 & (100,0,0) & (98,0,2) & (100,0,0) & (100,0,0) \\ 
		\hline
		\multicolumn{6}{c}{Model \uppercase\expandafter{\romannumeral2}, $\rho = 0.5$} \\
		8 & 1.1 & (0,99,1) & (4,96,0) & (64,36,0) & (98,2,0) \\ 
		10 & 1.38 & (4,95,1) & (39,61,0) & (96,4,0) & (100,0,0) \\ 
		12 & 1.65 & (26,73,1) & (82,18,0) & (100,0,0) & (100,0,0) \\ 
		15 & 2.07 & (80,19,1) & (100,0,0) & (100,0,0) & (99,0,1) \\ 
		18 & 2.48 & (96,3,1) & (100,0,0) & (100,0,0) & (100,0,0) \\ 
		22 & 3.03 & (99,0,1) & (100,0,0) & (100,0,0) & (100,0,0) \\  
		\hline
		\multicolumn{6}{c}{Model \uppercase\expandafter{\romannumeral2}, $\rho = 0.9$} \\
		11 & 1.04 & (1,98,1) & (8,90,2) & (63,37,0) & (91,9,0) \\ 
		13 & 1.23 & (8,91,1) & (21,77,2) & (83,17,0) & (99,1,0) \\ 
		16 & 1.51 & (26,72,2) & (69,29,2) & (98,2,0) & (100,0,0) \\ 
		21 & 1.98 & (75,23,2) & (96,2,2) & (100,0,0) & (100,0,0) \\ 
		26 & 2.45 & (96,2,2) & (98,0,2) & (100,0,0) & (100,0,0) \\ 
		32 & 3.02 & (98,0,2) & (98,0,2) & (100,0,0) & (100,0,0) \\  
		\hline
		\hline
	\end{tabular}
\end{table}

\bibliographystyle{apalike}
\bibliography{reference}